\documentclass[format=acmsmall, review=false]{acmart}

\makeatletter
\def\@ACM@checkaffil{% Only warnings
    \if@ACM@instpresent\else
    \ClassWarningNoLine{\@classname}{No institution present for an affiliation}%
    \fi
    \if@ACM@citypresent\else
    \ClassWarningNoLine{\@classname}{No city present for an affiliation}%
    \fi
    \if@ACM@countrypresent\else
        \ClassWarningNoLine{\@classname}{No country present for an affiliation}%
    \fi
}
\makeatother

\usepackage{acm-ec-23}
\usepackage{booktabs} 
\usepackage[ruled, linesnumbered]{algorithm2e} 

\SetAlFnt{\small}
\SetAlCapFnt{\small}
\SetAlCapNameFnt{\small}
\SetAlCapHSkip{0pt}
\IncMargin{-\parindent}

\setcitestyle{authoryear}

\usepackage{color}
\usepackage{tabularx}
\usepackage{amsmath}
\usepackage{amsthm}

\usepackage{apxproof}
\usepackage{mathtools}
\usepackage{amsfonts,textcomp}
\usepackage{latexsym}
\usepackage{graphicx} 
\usepackage{tikz}
\usepackage{multirow}
\usepackage{microtype}
\usepackage{rotating}
\usepackage{url}
\usepackage{color}
\usepackage{comment}
\usepackage{paralist}
\usepackage{booktabs}
\usepackage{pifont}

\usepackage[textsize=tiny]{todonotes}
\usepackage{cleveref}
\usepackage{bm}
\usepackage[]{units}
\usepackage{framed}
\usepackage{xspace}
\newtheorem{definition}{Definition}
\newtheorem{proposition}{Proposition}
\newtheorem{example}{Example}

\newtheorem{observation}{Observation}
\newtheorem{remark}{Remark}
\usepackage{thm-restate}
\newcommand{\ejrp}{EJR+\xspace}
\newcommand{\pjrp}{PJR+\xspace}
\newcommand{\rpjrp}{rank-\pjrp}
\newcommand{\clrstr}{20}

\DeclareMathOperator{\rank}{rank}

\DeclareMathOperator*{\argmin}{arg\,min}
\title{Robust and Verifiable Proportionality Axioms for Multiwinner Voting}

\author{Markus Brill}\affiliation{
\institution{University of Warwick}
\country{United Kingdom}
}
\email{Markus.Brill@warwick.ac.uk}
\author{Jannik Peters}
\affiliation{
\institution{Technische Universität Berlin}
\country{Germany}
}
\email{jannik.peters@tu-berlin.de}

\begin{abstract}
When selecting a subset of candidates (a so-called \textit{committee}) based on the preferences of voters, proportional representation is often a major desideratum.  When going beyond simplistic models such as party-list or district-based elections, it is surprisingly challenging to capture proportionality formally. As a consequence, the literature has produced numerous competing criteria of when a selected committee qualifies as proportional. Two of the most prominent notions are Dummett's \textit{proportionality for solid coalitions} (PSC) and Aziz et al.'s \textit{extended justified representation} (EJR). Both definitions guarantee proportional representation to groups of voters who have very similar preferences; such groups are referred to as \textit{solid coalitions} by Dummett and as \textit{cohesive groups} by Aziz et al. However, these notions lose their bite when groups are only almost solid or almost cohesive. 
In this paper, we propose proportionality axioms that are more robust than their existing counterparts, in the sense that they guarantee representation also to groups that do not qualify as solid or cohesive. Importantly, we show that these stronger proportionality requirements are always satisfiable. 
Another important advantage of our novel axioms is that their satisfaction can be easily verified: Given a committee, we can check in polynomial time whether it satisfies the axiom or not. This is in contrast to many established notions like EJR, for which the corresponding verification problem is known to be intractable. 

In the setting with approval preferences, we propose a robust and verifiable variant of EJR and a simply greedy procedure to compute committees satisfying it. We show that our axiom is considerably more discriminating in randomly generated instances compared to EJR and other existing axioms. 
In the setting with ranked preferences, we propose a robust variant of Dummett's PSC. In contrast to earlier strengthenings of PSC, our axiom can be efficiently verified even for general weak preferences. In the special case of strict preferences, our notion is the first known satisfiable proportionality axiom that is violated by the \textit{Single Tranferable Vote} (STV). In order to prove that our axiom can always be satisfied, we extend the notion of priceability to the ranked preferences setting. 
We also discuss implications of our results for participatory budgeting, querying procedures, and to the notion of proportionality degree.
\end{abstract}

\begin{document}

\begin{titlepage}

\maketitle

\end{titlepage}

\section{Introduction}

The proportional representation of preferences is an important goal in many scenarios in which a subset of candidates needs to be selected based on the preferences of voters over those candidates. Such scenarios occur in a wide variety of applications, including parliamentary elections \citep{Puke14a}, participatory budgeting \citep{PPSS21a}, 
digital democracy platforms \citep{BKNS14a},
and blockchain consensus protocols \citep{CeSt21a}.
In the (computational) social choice literature, this type of problem is often referred to as \textit{committee selection} or \textit{multiwinner voting} \citep{FSST17a,LaSk22a}. 
Some classic applications assume that candidates or voters (or both) come in predefined categories (political parties or voting districts), which greatly simplifies the task of finding representative outcomes. In the general case, when neither candidates nor voters come in predefined groups, it is surprisingly challenging to capture proportional representation formally. Perhaps as a consequence of this, the (computational) social choice literature has produced numerous competing criteria for when a selected committee qualifies as ``proportional.'' 

What many of the existing definitions have in common is that they define proportionality over groups of voters whose preferences are similar to each other. This approach goes back to the seminal work of \citet{Dumm84a}, who defined \textit{proportionality for solid coalitions (PSC)} in the setting where voters cast ranked ballots. 
PSC guarantees an appropriate level of representation to any group of voters that is ``solidly committed'' to a set of candidates in the sense that all voters of the group rank those candidates (in some order) over all other candidates. The most prominent example of a voting rule ensuring PSC is the widely used \textit{single transferable vote (STV)}.\footnote{In his article on STV, \citeauthor{Tide95a} remarked that ``it is the fact that STV satisfies PSC that justifies describing STV as a system of proportional representation'' \citep[][page~27]{Tide95a}.}
Similar notions were subsequently introduced in the setting of approval-based multiwinner voting \citep{LaSk22a}. In particular, extended justified representation (EJR) \citep{ABC+16a} and proportional justified representation (PJR) \citep{SFF+17a} formulate proportional representation guarantees for ``cohesive'' groups; a group of voters qualifies as cohesive if the intersection of their approval sets is sufficiently large.

When voters with similar preferences fall short of the high standard of uniformity defined by ``solid coalitions'' or ``cohesive groups,'' the axioms stay mostly mute.\footnote{PSC does not impose any lower bounds on the representation of an almost solid group. EJR, on the other hand, does at least impose weakened representation guarantees for less cohesive groups \citep{SFF+17a}.}
Indeed, this reliance on highly uniform voter groups has attracted criticism in the literature.  
For instance, \citeauthor{Tide06a} remarked (in the context of discussing a rule satisfying PSC) that there may be ``voters who would be members of a solid coalition except that they included an `extraneous' candidate, which is quickly eliminated among their top choices. These voters’ nearly solid support for the coalition counts for nothing, which seems to me inappropriate'' \citep[page~279]{Tide06a}. \citeauthor{AzLe20a} gave a concrete example for this behavior and stated\,---\,with regard to their own Expanding Approvals Rule (EAR)\,---\,that ``understanding formally whether EAR, or other rules, satisfy Tideman’s notion of ‘robust’ PSC is an interesting avenue for future work'' \citep[page~33]{AzLe20a}.  
Relatedly, \citet{HKRW21} criticize that PSC is not compatible with ballot truncation. For instance, in an election where two candidates are to be elected and one quarter of the voters only rank $a$ whereas another quarter of the voters only rank $b$ before $a$, PSC would not require either $b$ or $a$ to be elected. 
This is further corroborated by the work of \citet{MarPle16a}, who find frequent cases of vote splitting in Irish STV elections, which has the potential to make large solid coalitions quite rare.

Similar empirical criticism was also voiced for the justified-representation axioms in approval-based multiwinner voting. For instance, \citet{BFNK19a} find that large cohesive groups do not seem to be very common in their experiments and that, for the preference models they studied, even a randomly chosen committee satisfies EJR and PJR with non-negligible probability. 
A similar effect was noticed by \citet{SFJ+22a}, who noted that the more ``realistic'' of their statistical models seem to have a low ``cohesiveness level.'' 

An unrelated criticism of proportionality notions such as EJR and PJR is that they cannot be verified in polynomial time: It is coNP-complete to check whether a given committee satisfies EJR \citep{ABC+16a} or PJR \citep{AEH+18a}. This is a crucial downside in applications in which the proportionality of the outcome needs to be verifiable \cite[e.g.,][]{CeSt21a,MSW22a}. The same criticism applies to \citeauthor{AzLe20a}'s  generalization of PSC to weak preferences (i.e., rankings containing ties): it is coNP-complete to check whether a given committee satisfies the axiom \citep[Proposition~13]{AzLe20a}. PSC itself, which is only defined for the special case of strict preferences (i.e., rankings without ties), is verifiable in polynomial time.

\subsection{Our Contribution}

In this paper, we propose novel proportionality axioms that address the criticisms described above. Our axioms are (1) \textit{robust} in the sense that they guarantee proportional representation also to voter groups that do not qualify as ``solid'' or ``cohesive'' and (2) \textit{verifiable} in the sense that it can be checked in polynomial time whether a given committee satisfies the axiom or not. 
Our axioms are more demanding than existing ones, as they impose strictly more constraints on committees. Importantly, however, we show that these stronger proportionality requirements can be satisfied in all instances. Indeed, we identify voting rules from the literature that always produce committees satisfying our strong requirements. Our results can, therefore, be interpreted as evidence that those rules satisfy proportionality to a high extent.    
For an overview of the proportionality axioms considered in this paper, we refer to \Cref{fig:relations} on page~\pageref{fig:relations}.

In the setting with approval preferences, we propose \textit{\ejrp} as a robust and verifiable strengthening of EJR, together with a simply greedy procedure to compute committees satisfying \ejrp. Using randomly generated preference profiles, we demonstrate that \ejrp is a considerably more demanding axiom compared to EJR and other existing axioms. We also observe that established rules such as \textit{Proportional Approval Voting (PAV)} and the \textit{Method of Equal Shares (MES)} satisfy \ejrp and that \ejrp can be\,---\,in contrast to EJR and PJR\,---\,efficiently verified.

In the setting with ranked preferences, we propose \textit{rank-\pjrp} as a robust strengthening of Dummett's PSC. In contrast to earlier strengthenings of PSC, rank-\pjrp can be efficiently verified even for general weak preferences. We observe that STV violates rank-\pjrp. To the best of our knowledge, this establishes rank-\pjrp as the first satisfiable proportionality axiom that separates STV from more sophisticated methods such as the expanding approvals rule (EAR).\footnote{Earlier strengthenings of PSC that are violated by STV are either sometimes unsatisfiable \citep{AEFLS17} or equivalent to PSC in the case of strict preferences (for which STV is defined) \citep{AzLe20a,AzLe21a}.}
In order to prove that rank-\ejrp can always be satisfied, we extend the notion of priceability \citep{PeSk20a} to the ranked preferences setting and show that EAR satisfies it. Moreover, we use randomly generated preference profiles to show that rank-\pjrp is much more demanding than~PSC.

Finally, we extend our robustness approach to the proportionality degree \citep{Skow21a} and to two applications that are closely related to multiwinner voting:  participatory budgeting \citep{PPS21a} and querying procedures for civic participation platforms \citep{HKP+23a}. 

Our paper treats approval preferences and ranked preferences within a unified framework and establishes novel relationships between approval-based and ranking-based axioms. Hence, our work helps to consolidate the literature from the approval-based and ranking-based model, a task explicitly encouraged by \citet[pages~95--96]{LaSk22a}.

\subsection{Related Work} 

The study of proportional representation in multiwinner voting has a long tradition and voting rules aiming to produce proportional committees have been proposed long before the first proportionality notions have been formalized (see, e.g., the historical notes in the surveys by \citet{Tide95a}, \citet{MMM96a}, and \citet{Jans16a}). 
For ranked preferences, the first and most well-known formal proportionality axiom is the aforementioned \emph{proportionality for solid coalitions (PSC)}, which was introduced by eminent philosopher Sir Michael Dummett 
\citep{Dumm84a}.  
Extensions of PSC to weak rankings were only recently introduced by \citet{AzLe20a,AzLe21a}, who also provided a characterization of committees satisfying PSC \citep{AzLe22a}. 
\citet{AEFLS17} discussed several extensions or variants of Condorcet consistency to multiwinner voting. One of their notions, local stability, was independently studied by \citet{JMW20a}. 

In \textit{approval-based} multiwinner voting \citep{LaSk22a}, proportionality axioms have received a lot of attention in recent years. Starting with the work of \citet{ABC+16a}, who introduced not only EJR but also \textit{core stability}, several papers either generalized these proportionality notions, found new rules satisfying them, or identified new settings to apply them. For instance, PJR was introduced by \citet{SFF+17a} and subsequently studied by \citet{BFJL16a} and \citet{AEH+18a},
the \textit{proportionality degree} was introduced by \citet{Skow21a} and further studied by \citet{JaFa22a},
\textit{individual representation} was introduced by \citet{BIMP22a}, 
and \textit{fully justified representation} was introduced by \citet{PeSk20a}, who also proposed the \textit{Method of Equal Shares} and the concept of priceability. 

Commonly studied formalisms that are closely related to approval-based multiwinner voting include 
 \textit{participatory budgeting} \citep{PPS21a, LCG22a, BFL+23a, ALT18a},
 \textit{proportional rankings} \citep{SLB+17a, isbr21b, RST22a}, 
and \textit{public decision-making} \citep{SkGo22a, FKP21a}. 
In many cases, axioms like EJR and PJR and voting rules like PAV and MES have been adapted to these related settings.   
Interestingly, \citet{SkGo22a} motivate their axioms with the goal to ``[...] guarantee fair treatment for all groups of voters, not only the cohesive ones.'' Since they work in a setting with multiple binary issues, their concepts and results do not translate to the multiwinner setting we study. 

Finally, a recent line of work studying approximations of core stability in approval-based multiwinner voting and beyond \citep{CJMW19a, PeSk20a, JMW20a, MSWW22a}. Determining whether the core of an approval-based multiwinner election is always nonempty is considered an important open question~\citep{LaSk22a}.

\section{Preliminaries}
\label{sec:prelims}

In this section, we formally introduce the setting and review proportionality axioms and voting rules from the literature. For a natural number $n$, let $[n]$ denote the set $\{1,\dots,n\}$. 

\subsection{Setting}

We consider a social choice setting with a finite set $C = \{c_1, \dots, c_m\}$ of $m$ \textit{candidates} and a finite set $N = [n]$ of \emph{voters} who have ordinal preferences over the candidates. Throughout this paper, we assume that preferences have the following form: For each voter  $i \in N$, there is a set $A_i \subseteq C$ of \textit{acceptable candidates} and a complete and transitive preference relation $\succeq_i \, \subseteq A_i \times A_i$ over the acceptable candidates. In other words, $\succeq_i$ is a \textit{weak order} over $A_i$. We let $\succ_i$ denote the strict part of~$\succeq_i$. We assume that voters strictly prefer acceptable candidates to unacceptable ones, and that they are indifferent among unacceptable candidates.

For $A, B \subseteq C$, we write $A \succeq_i B$ (respectively, $A \succ_i B$) if $a \succeq_i b$ (respectively, $a \succ_i b$) holds for all $a \in A$ and $b \in B$.
Further, for any $c \in A_i$ we let $\rank(i, c) = \lvert \{c' \in C \colon c' \succ_i c\} \rvert + 1$ denote the \emph{rank} voter $i$ assigns to candidate $c$. We say that voter $i$ ranks candidate $c$ \textit{higher} than candidate $c'$ if $c \succ_i c'$, or, equivalently, $\rank(i,c)<\rank(i,c')$. All unacceptable candidates $c \in C\setminus A_i$ are assigned a rank of $\rank(i,c)=+\infty$.

Besides the general case of weak-order preferences, we consider two important special cases. 
If~$\succeq_i$ is a linear order over $A_i$, we say that voter $i$ has \textit{strict} preferences. In this case, there are no ties between acceptable candidates. 
If, on the other hand, a voter is indifferent among all candidates in $A_i$, we say that the voter has \textit{dichotomous} preferences. Dichotomous preferences naturally occur when using approval ballots, which is why we also refer to them as \textit{approval preferences}.   
We use the term \textit{weak preferences} to refer to the general case, i.e., when preferences are not assumed to be strict or dichotomous.

A \textit{preference profile} $P=(\succeq_1, \dots, \succeq_n)$ contains the preferences of all voters. Note that, for each voter $i$, the set $A_i$ can be deduced from $\succeq_i$. 
We call a preference profile \textit{strict} if all voters have strict preferences. If all voters have dichotomous preferences, we refer to $P$ as an \textit{approval profile} and denote it as $P=(A_1, \dots, A_n)$. For a given approval profile and a candidate $c \in C$, we let  $N_c = \{i \in N \colon c \in A_i\}$ denote the set of approvers of $c$.

We often write the preferences of voters as a strict ranking over indifference classes, omitting unacceptable candidates. 
The following example illustrates this. 

\begin{example}
    Consider the following preference profile with $m=6$ candidates 
    and $n=3$ voters.
    \begin{align*}
    &1: c_1 \succ \{c_2, c_3, c_4\} \\
    &2: \{c_2, c_3\} \\
    &3: c_5 \succ c_4 \succ c_3
    \end{align*}
Here, we have $A_1=\{c_1, c_2, c_3, c_4\}$, $A_2=\{c_2, c_3\}$, and $A_3=\{c_3, c_4, c_5\}$. The ranks that voter 1 assigns to the candidates are given by $\rank(1,c_1)=1$, $\rank(1,c_2)=\rank(1,c_3)=\rank(1,c_4)=2$, and $\rank(1,c_5)=+\infty$. 
Voter~2 has dichotomous preferences and voter~3 has strict preferences.  
\end{example}

A \textit{(multiwinner voting) instance} consists of a set $N$ of voters, a set $C$ of candidates, a preference profile $P$, and target committee size $k\le m$. A \emph{feasible committee} is any subset $W \subseteq C$ with $\lvert W \rvert \le k$. 
A \textit{(multiwinner voting) rule} maps every instance $(N,C,P,k)$ to a non-empty set of feasible committees. We allow a rule to output more than one committee to account for ties and in order to be able to speak about rules such as EAR and STV (see \Cref{sec:rules}) that come in several different variants. 
We say that a rule ``satisfies'' a proportionality notion if and only if, for each instance, \textit{every} committee in the output of the rule satisfies the respective notion.

\subsection{Proportionality Notions for Strict Preferences} 

We now turn to the proportionality notions defined in the literature, starting with the oldest and most prominent setting: multiwinner elections with strict preferences. In his classical work, \citet{Dumm84a} introduced the notion of \emph{Proportionality for Solid Coalitions (PSC)}. To define this property, we first need to define the eponymous solid coalitions.

\begin{definition}[Solid Coalition] \label{def:SC}
Given a strict preference profile, a subset $N' \subseteq N$ of voters forms a \emph{solid coalition} over a set of candidates $C' \subseteq C$ if $C' \succ_i C\setminus C'$ for all $i \in N'$.
\end{definition}

Hence, voters in a solid coalition rank all candidates in $C'$ higher than candidates outside of $C'$, but the order among candidates in $C'$ may differ among voters in the coalition. 
Since the group~$N'$ contains an $(|N'|/n)$-fraction of all voters, 
PSC requires that at least $\lfloor (|N'|/{n}) k \rfloor$ candidates from this prefix are selected.\footnote{In accordance with the literature on approval-based committee voting, our definition of PSC is based on the so-called \textit{Hare quota} $\frac{n}{k}$. Different choices of quota are often discussed;  e.g., the \textit{Droop quota} is given by 
$\frac{n}{k+1}$ \citep{AzLe20a}.}

\begin{definition}[PSC] \label{def:PSC}
Given an instance with strict preferences, a feasible committee $W$ satisfies \emph{proportionality for solid coalitions (PSC)} if for any subset $N' \subseteq N$ of voters forming a solid coalition over $C' \subseteq C$ and any $\ell \in \mathbb{N}$ such that $\lvert N' \rvert \ge \ell \frac{n}{k}$ it holds that 
$\lvert C' \cap W\rvert \ge \min\left(\lvert C'\rvert,\ell\right).$
\end{definition}

\subsection{Proportionality Notions for Approval Preferences}
\label{sec:prelim-approval}

Inspired by \citet{Dumm84a}, in the setting with approvalpreferences, \citet{ABC+16a} and \citet{SFF+17a} introduced their \emph{justified-representation} axioms. Just as solid coalitions are the foundation for PSC, the justified-representation axioms build on cohesive groups. 

\begin{definition}[Cohesive Group]
 Given an approval profile and a natural number $\ell$, a subset $N' \subseteq N$ of voters  forms an $\ell$-cohesive group if $\lvert N' \rvert \ge \ell \frac{n}{k}$ and $\big \lvert \bigcap_{i \in N'} A_i \big \rvert \ge \ell$.
\end{definition}

This now serves as the main ingredient for the two most prominent justified-representation notions: 
\textit{extended justified representation}~\citep{ABC+16a} and
\textit{proportional justified representation}~\citep{SFF+17a}. 

\begin{definition}[EJR \normalfont{\&} PJR] \label{def:pjr-ejr}
Given an instance with approval preferences, a feasible committee~$W$ satisfies
\begin{itemize}
    \item \emph{extended justified representation (EJR)} if for each natural number $\ell$ and for each $\ell$-cohesive group $N'\subseteq N$, there is some voter $i \in N'$ with $\big \lvert A_i \cap W \big \rvert \ge \ell$; and 
    \item \emph{proportional justified representation (PJR)} if for each natural number $\ell$ and for each $\ell$-cohesive group $N' \subseteq N$, it holds that $\big \lvert \bigcup_{i \in N'} A_i \cap W \big \rvert \ge \ell$.
\end{itemize}
\end{definition}

By definition, EJR is a stronger requirement than PJR. 
Finally, we introduce the concept of \textit{priceability} \citep{PeSk20a, PPS21a}. 

\begin{definition}[Priceability] \label{def:priceability}
Given an instance with approval preferences, a committee $W$ is \emph{priceable} if there exist a $B>0$ and functions $p_i\colon C \rightarrow [0, \frac{B}{n}]$ such that the following conditions hold: 
\begin{itemize}
    \item[\textbf{(C1)}] $c \notin A_i \implies p_i(c) = 0$ 
    \item[\textbf{(C2)}] $\sum_{c \in C} p_i(c) \le \frac{B}{n}$ for all $i \in N$
    \item[\textbf{(C3)}] $\sum_{i \in N} p_i(c) = 1$ for all $c \in W$
    \item[\textbf{(C4)}] $\sum_{i \in N} p_i(c) = 0$ for all $c \notin W$
    \item[\textbf{(C5)}] ${\sum_{i \in N \colon c \in A_i} \left(\frac{B}{n} - \sum_{c' \in C}p_i(c')\right)}% 
    \le 1$ for all $c \in C \setminus W$
\end{itemize}
In this case, the pair $\{B, \{p_i\}_{i \in N}\}$ is called a \emph{price system} for $W$. 
\end{definition}
\citet{PeSk20a} have shown that every priceable committee of size $k$ satisfies PJR.

\subsection{Proportionality Notions for Weak Preferences}
\label{sec:prelims-weak}

For committee elections with general weak preferences, \citet{AzLe20a} generalized the notion of solid coalitions (\Cref{def:SC}) and PSC (\Cref{def:PSC}) in the following way.

\begin{definition}[Generalized Solid Coalition] \label{def:GSC}
Given a preference profile, a group $N' \subseteq N$ of voters forms a generalized solid coalition over a set $C'\subseteq C$ of candidates if $C' \succeq_i C\setminus C'$ for all $i \in N'$.
\end{definition}

In order to qualify as a generalized solid coalition, no candidate outside of $C'$ can be ranked higher than any candidate inside $C'$. Ties, however, are allowed.

To facilitate the definition of generalized PSC, we define the \emph{upper contour set} of a set $C'$ of candidates w.r.t. a set $N'$ of voters as the set consisting of those candidates $c$ for which there is at least one voter in $N'$ that does not rank all candidates in $C'$ strictly higher than $c$.
Formally, the upper contour set $\overline{C'}(N')$ of $C'\subseteq C$ w.r.t. $N'\subseteq N$ is defined as  
\[
\overline{C'}(N') = \{c \in C \colon  \text{there exist } i \in N' \text{ and } c' \in C' \text{ such that } c \succeq_i c'\}.
\]

\begin{definition}[Generalized PSC] \label{def:GSC}
Given an instance with weak preferences, a  feasible committee~$W$ satisfies \emph{generalized PSC} if for every natural number $\ell$ and for every group $N'\subseteq N$ of voters forming a generalized solid coalition over $C'$ with with $|N'| \ge \ell \frac{n}{k}$, it holds that 
\[\lvert \overline{C'}(N') \cap W\rvert \ge \min(\ell, \lvert C'\rvert). \]
\end{definition}

Generalized PSC is equivalent to PJR when restricted to instances with approval preferences and\,---\,by definition\,---\,equivalent to PSC when restricted to instances with strict preferences. The former equivalence was proven by \citet{AzLe20a}, who also introduced a rule that satisfies generalized PSC for general weak preferences (see \Cref{sec:rules}). 

In a follow-up work on participatory budgeting, \citet{AzLe21a} introduced a slight strengthening of generalized PSC called \emph{Inclusion PSC (IPSC)}, which we reformulate here for our setting.

\begin{definition}[IPSC] \label{def:ipsc}
Given an instance with weak preferences, a feasible committee~$W$ satisfies \emph{Inclusion PSC (IPSC)} if for every natural number $\ell$ and for every group $N'\subseteq N$ of voters forming a generalized solid coalition over $C'$ with $|N'| \ge \ell \frac{n}{k}$, it holds that
\[
C' \subseteq W \quad \text{or} \quad \lvert \overline{C'}(N') \cap W\rvert \ge \ell. 
\]
\end{definition}

The following example illustrates the difference between generalized PSC and IPSC. 

\begin{example}
Consider the following instance with $n=3$ voters, $m=9$ candidates, and $k=6$.
\begin{align*}
    &1: c_1 \succ \{c_2, c_3, c_4\} \succ c_7 \succ c_8 \\
    &2: c_1 \succ \{c_2, c_5, c_6\} \succ c_8 \succ c_7 \\
    &3: c_7 \succ c_8 \succ c_9
\end{align*}
Each voter $i\in \{1,2,3\}$ by herself forms a generalized coalition over the candidates with $\rank(i,c)\le 3$ and deserves to be represented by $\frac{1}{3} k=2$ candidates. 
Moreover, voters 1 and 2 together form a generalized solid coalition over $\{c_1, c_2\}$ of size $4 \frac{n}{k}$.
Generalized PSC prescribes that $c_1, c_7, c_8$, at least one of $\{c_2, c_3, c_4\}$, and at least one of $\{c_2, c_5, c_6\}$ needs to be selected.  
Therefore, the committee $W = \{c_1, c_3, c_5, c_7, c_8, c_9\}$ satisfies generalized PSC. On the other hand,  $W$ does not satisfy IPSC since the group $N' = \{1,2\}$ solidly supports $C' = \{c_1, c_2\}$ with $c_2 \notin W$ and
$\lvert \overline{C'}(N') \cap W\rvert = \lvert \{c_1, c_3, c_5\} \rvert < 4$.
\end{example}

We note that for strict preferences, IPSC and generalized PSC coincide with Dummett's PSC.
\begin{restatable}{observation}{obsequiv}
    For strict preferences, both IPSC and generalized PSC are equivalent to PSC.
    \label{obs:ipsc_psc}
\end{restatable}
\begin{proof}
    It is already known that IPSC implies generalized PSC and generalized PSC implies PSC. 
    Consider an instance with strict preferences and assume that a committee $W$ does not satisfy IPSC. Then there is a group of voters $N'$ forming a solid coalition over $C'$ such that less than $\ell$ of $\overline{C'}(N')$ and not all of $C'$ are in $W$. Since the preferences are strict $\overline{C'}(N') =C'$ has to hold. Hence, this is precisely, the requirement for PSC and thus PSC is also violated.
\end{proof}

\citet{AzLe20a} have shown that it is coNP-complete to verify whether a given committee satisfies generalized PSC. We show that the same is true for IPSC.

\begin{restatable}{proposition}{ipscconp}\label{prop:checkingIPSC}
    Given an instance with weak preferences and a feasible committee $W$, it is coNP-complete to check whether $W$ satisfies IPSC.
\end{restatable}
\begin{proof}
    First, we note that membership in coNP follows from the fact that if there is a violation to IPSC, one can use the candidates $C'$ and group $N'$ of voters as a witness to this violation, since it can be verified in polynomial time whether $N'$ indeed form a generalized solid coalition and whether there are sufficiently many candidates selected from their upper contour set.
    
    For the hardness, we reduce from the \textsc{Clique} problem. Given a graph $G = (V,E)$ and an integer~$k'$, the task is to determine whether there are $k'$ vertices in $V$ forming a clique, i.e., a complete subgraph. Without loss of generality, we assume $\lvert V \rvert$ to be divisible by $k'$; this can easily be achieved by adding isolated vertices. Since \textsc{Clique} is NP-complete, its complement is coNP-complete.

    For a given \textsc{Clique} instance $(G,k)$, we construct a multiwinner voting instances as follows. Let $v_1, \dots, v_n$ be the vertices of $G$. For each vertex $v_i$, there is one voter $i$ and one candidate $c_i$.  Further, let $N(i)$ denote the set of candidates corresponding to the neighbors of $v_i$ in $G$. The preferences of voter $i \in [n]$ are given by 
    \[
    i: c_i \succ N(i)
    \]
    with all candidates in $N(i)$ being in a tie.
    We set the target committee size to $k = \frac{n}{k'}$ and add $k$ dummy candidates $d_1, \dots, d_k$.
    We claim that the committee $W=\{d_1, \dots, d_k\}$ consisting of all dummy candidates satisfies IPSC if and only if there is no clique of size $k'$ in $G$.

    First, assume that $W$ does not satisfy IPSC. Then there must be a subset $N'$ of voters  solidly supporting a set of candidates $C'$ such that $\lvert N' \rvert \ge \frac{n}{k} = k'$. Since, the unique first choice of each voter is the candidate corresponding to themselves, all candidates corresponding to voters in $N'$ must be in $C'$. Otherwise, they would not form a solid coalition, since there would be a voter missing their first choice candidate. Then, however, every vertex corresponding to a voter in $N'$ must be connected to every other vertex, since they rank them at least at second place. Therefore, the vertices form a clique of size $k'$.

    On the other hand, if we have a clique of size $k'$, the voters corresponding to the clique form a solid coalition of size $\frac{n}{k}$ over the candidates corresponding to the clique, thus witnessing a violation of IPSC.
\end{proof}

\subsection{Multiwinner Voting Rules}
\label{sec:rules}

Finally, we introduce four rules that we consider throughout the paper: EAR, MES, PAV, and STV. While the first two are defined for general weak preferences, PAV assumes approval preferences and STV assumes strict preferences. 
To emphasize the similarity between EAR, MES, and STV, we formulate these rules in a slightly non-standard way. 

\begin{algorithm}
\caption{Expanding Approvals Rule}
\label{alg:ear}

$W \gets \emptyset$\;
 $b_i \gets \frac{k}{n}$ for all $i \in N$\;
\For{$r$ in $1, \dots, m$}{
 $N_c^r \gets \{i \in N \colon \rank(i,c) \le r\}$\;
 $X \gets \{c \in C \setminus W \colon \sum_{i \in N_c^r} b_i \ge 1 \}$\;
\While{$X \neq \emptyset$}{
 select $c \in X$\label{line:select1}\;
 $W \gets W \cup \{c\}$\;
 define $(\delta_i)_{i \in N_c^r}$ such that $\sum_{i \in N_c^r} \delta_i =1$ and $0\le \delta_i \le b_i$ for all $i \in N_c^r$\label{line:select2}\;
\For{$i \in N_c^r$}{
 $b_i \gets b_i -\delta_i$\;
}
 $X \gets \{c \in C \setminus W \colon \sum_{i \in N_c^r} b_i \ge 1 \}$\;
}
}
 return $W$
\end{algorithm}

The \emph{expanding approvals rule (EAR)} is a family of rules introduced by \citet{AzLe20a, AzLe21a} for the general setting with weak preferences. Intuitively, the family consists of rules which give each voter a budget of $\frac{k}{n}$ with the goal to buy candidates for a price of $1$. The rules go rank-by-rank and check, for every rank $r$, if there is a candidate who can be afforded by the voters who assign a rank of at most $r$ to this candidate. If there is, one of these candidates is selected and the budget of the supporting voters is decreased by $1$ (the price of the candidate). See \Cref{alg:ear} for an algorithmic template of these rules.
Different instantiations of this family differ in \textit{(i)} how they select the candidate to add (line~\ref{line:select1})
and \textit{(ii)} how they decrease the budgets of the voters approving this candidate (line~\ref{line:select2}).

The \emph{Method of Equal Shares (MES)} can be described as a special case of EAR. It was first introduced for instances with approval preferences by \citet{PeSk20a} (who referred to it as ``Rule~X'') and subsequently generalized to participatory budgeting with additive utilities as well as strict rankings by \citet{PPS21a}. In their variant of MES, \citet{PPS21a} define it using their framework of additive utilities by assuming that the utilities are lexicographic. 
In our paper, we use a simplified version of the rule. 
To define MES, for a given $\rho > 0$ and budgets $(b_i)_{i \in N}$, we say that a candidate is $\rho$-affordable if $\sum_{i \in N_c^r} \min(\rho, b_i) = 1$. Here, $N_c^r$ denotes the set of voters who assign to $c$ a rank of at most $r$. 
MES is the variant of EAR which always selects the candidate which is $\rho$-affordable for the lowest $\rho$ and then reduces the budget $b_i$ for each $i \in N_c^r$ by $\delta_i = \min(\rho, b_i)$.

A well-known rule for instances with approval preferences is \emph{proportional approval voting (PAV)} \citep{Thie95a}. 
PAV selects all committees $W$ of size $|W|=k$ maximizing $\sum_{i \in N}\sum_{j = 1}^{\lvert A_i \cap W\rvert} \frac{1}{j}$.

Finally, we consider the widely-used \emph{Single Transferable Vote (STV)}  family of voting rules \citep{Tide95a}, 
which assumes that preferences are strict. 
We again work with the Hare quota $\frac{k}{n}$. STV can be described similarly to EAR: Each voter is assigned a budget of $\frac{k}{n}$ in the beginning, with the price of a candidate being $1$. The crucial difference to EAR is that instead of going rank-by-rank, STV checks if there is a candidate who can be afforded by the voters ranking it \textit{first}. If there is no such candidate, a candidate with the lowest amount of budget among voters ranking it first is eliminated. Different versions of STV differ in deciding \textit{(i)} which candidate to eliminate, \textit{(ii)} which affordable candidate to select, and \textit{(iii)} how to transfer the surplus budget. For an algorithmic description, we refer to \cref{alg:stv}.

\begin{algorithm}
\caption{Single Transferable Vote}
\label{alg:stv}

 $W \gets \emptyset$\;
 $b_i \gets \frac{k}{n}$ for all $i \in N$\;
\While{$C\neq \emptyset$}{
 $N_c \gets \{i \in N \colon \rank(i,c) > \rank(i,c') \text{ for all } c' \in C \setminus \{c\}\}$\;
 $X \gets \{c \in C \setminus W \colon \sum_{i \in N_c} b_i \ge 1 \}$\;
\While{$X \neq \emptyset$}{
 select $c \in X$ \label{line:select1a}\;
 $W \gets W \cup \{c\}$\;
 define $(\delta_i)_{i \in N_c^r}$ such that $\sum_{i \in N_c^r} \delta_i =1$ and $0\le \delta_i \le b_i$ for all $i \in N_c^r$\label{line:select2a}\;
\For{$i \in N_c^r$}{
 $b_i \gets b_i -\delta_i$\;
}
$C \gets C \setminus \{c\}$\;
 $N_c \gets \{i \in N \colon \rank(i,c) > \rank(i,c') \text{ for all } c' \in C \setminus \{c\}\}$\;
 $X \gets \{c \in C \setminus W \colon \sum_{i \in N_c} b_i \ge 1 \}$\;
}
 select $c \in \argmin_{c' \in C} \sum_{i \in N_c} b_i$\;
 $C \gets C \setminus \{c\}$\;
}
 return $W$\;
\end{algorithm}

\section{Approval Ballots}
\label{sec:approval}

We first consider the case of approval preferences. Our notions for ranked preferences in \Cref{sec:ranked} will build on the notions we develop here.
We consider a robust strengthening of PJR in \Cref{sec:pjr+}, a robust strengthening of EJR in \Cref{sec:ejr+}, and a simple greedy algorithm satisfying the robust EJR notion in \Cref{sec:greedy}. In \Cref{sec:exp-app}, we show that the robust EJR notion is considerably more discriminating than existing proportionality axioms for randomly generated instances.  

 Our point of departure is the following observation: Existing proofs establishing that certain voting rules satisfy justified-representation axioms like PJR and EJR do \textit{not} actually use the fact that the voter group under consideration forms a cohesive group. (For example, this is true for the proof by \citet{ABC+16a} that PAV satisfies EJR  and for the proof by \citet{PeSk20a} that MES satisfies EJR.) 
 Instead, these proofs only argue about the existence of certain candidates that are not included in the committee but approved by the voter group. 
Based on this observation, we can reformulate justified-representation axioms in a robust way,  moving away from arguing about cohesive groups and rather focusing on candidates outside the committee who have a valid ``proportional'' reason to be included. Adapting proportionality requirements in this way means that voter groups are no longer required to be highly cohesive in order to be properly represented in the committee. Rather, it is sufficient to identify an unselected candidate that is approved by all voters in the group and whose addition would bring the group closer to their justified representation. 
 
 \subsection{Proportional Justified Representation (PJR) Without Cohesiveness}
 \label{sec:pjr+}

 Before defining our main new axiom in \Cref{sec:ejr+}, we prepare the ground by reformulating (and renaming) an existing notion: \textit{Inclusion PSC}. IPSC was originally defined for instances with general weak preferences (see \Cref{def:ipsc}). When restricted to instances with approval preferences, \citet{AzLe21a} observed that IPSC can be formulated as follows. 
 
 \begin{observation}[\citet{AzLe21a}, Proposition 9]
 Given an instance with approval preferences, a feasible committee~$W$ satisfies IPSC if and only if there is no group $N'\subseteq N$ of voters so that 
 \begin{itemize}
     \item[(i)]  $\lvert N' \rvert \ge (\left\lvert \bigcup_{i \in N'} A_i \cap W \right\rvert + 1)\frac{n}{k}$, and
     \item[(ii)] there exists some $c^* \in \bigcap_{i \in N'} A_i \setminus (\bigcup_{i \in N'} A_i \cap W)$.
 \end{itemize}
 \end{observation}

\citet{AzLe21a} have also shown that, for approval preferences, IPSC is strictly stronger than PJR (and incomparable to EJR). 
We note that IPSC can be reformulated in a way that better aligns this with the definitions of axioms in \Cref{sec:prelims}.  

 Before defining our main new axiom in \Cref{sec:ejr+}, we prepare the ground by reformulating (and renaming) an existing notion: \textit{Inclusion PSC}. IPSC was originally defined for instances with general weak preferences (see \Cref{def:ipsc}). When restricted to instances with approval preferences, \citet{AzLe21a} have shown that IPSC is strictly stronger than PJR (and incomparable to EJR). 
 Based on an earlier reformulation \citep[Proposition~9]{AzLe21a},
we note that IPSC can be reformulated in a way that aligns with the definitions in \Cref{sec:prelims}.  
 
\begin{observation}
Given an instance with approval preferences, a committee $W$ satisfies IPSC if and only if there is no candidate $c \in C \setminus W$, group of voters $N'\subseteq N$, and $\ell \in \mathbb{N}$ with $\lvert N'\rvert \ge \ell \frac{n}{k}$ such that
\[
c \in \bigcap_{i \in N'} A_i \quad \text{and} \quad \left\lvert \bigcup_{i \in N'} A_i \cap W \right\rvert < \ell.
\]
\label{obs:ipsc_abc}
\end{observation}
\vspace{-1em}
Thus, if a committee does not satisfy IPSC, there is an unselected candidate $c$ who has a claim to be a part of the committee, since there is a group approving $c$ which together do not get their ``fair share'' of $\ell$ candidates in the committee. Note that this group of voters is not required to be $\ell$-cohesive; the merely need to agree on candidate $c$.

In light of this observation, the name ``Inclusion-PSC'' is a misnomer when considering approval preferences. From now on, motivated by the fact that the property can be considered a robust strengthening of proportional justified representation~(PJR), we will refer to this property as \pjrp.

\begin{definition}[\pjrp]\label{def:pjrp}
    Given an instance with approval preferences, a committee satisfies \emph{\pjrp} if and only if it satisfies IPSC. 
\end{definition}

This new name will also make it easier to recognize that we are applying the IPSC axiom to the special case of approval preferences. 

We show that \pjrp can be verified in polynomial time by reducing the problem to submodular optimization (using the same reduction that \citet{BGP+19a} used in the approval-based apportionment setting). This is in contrast to PJR, which is coNP-complete to verify \citep{AEH+18a}.
\vspace{-1.5em}
\begin{restatable}{proposition}{pjrppoly}
    Given an instance with approval preferences and a feasible committee $W$, it can be verified in polynomial time whether $W$ satisfies \pjrp.
\end{restatable}
\begin{proof}
    Consider an instance $(N,C,P,k)$ with approval preferences, a feasible committee $W$, and a candidate $c \notin W$. Define the function $f\colon 2^{N_c} \to \mathbb{Q}$ as 
    \[f(N') = \left\lvert \bigcup_{i \in N'} A_i \cap W \right\rvert - \frac{\lvert N'\rvert k}{n}.\] 
    To show that this function is submodular, we need to show that for any $N'' \subset N' \subset N_c$ and $i \in N_c \setminus N'$ it holds that $f( N'' \cup \{i\}) - f(N'') \ge f( N' \cup \{i\}) - f(N')$. To see that this holds, we observe 
    \begin{align*}
        f( N'' \cup \{i\}) - f(N'')  &= \left\lvert \bigcup_{j \in N''} (A_j \cap W) \cup (A_i \cap W)\right\rvert - \left\lvert \bigcup_{j \in N''} A_j \cap W\right\rvert - \frac{k}{n} \\ 
        &= \left\lvert (A_i \cap W) \setminus \bigcup_{j \in N''} A_j \cap W \right\rvert - \frac{k}{n}
        \ge \left\lvert (A_i \cap W) \setminus \bigcup_{j \in N'} A_j \cap W \right\rvert - \frac{k}{n}
        \\ &= \left\lvert \bigcup_{j \in N'} (A_j \cap W) \cup (A_i \cap W)\right\rvert - \left\lvert \bigcup_{j \in N'} A_j \cap W\right\rvert - \frac{k}{n} \\
        &= f( N' \cup \{i\}) - f(N').
    \end{align*}
    \noindent Further, observe that a subset $N' \subseteq N_c$ witnesses a violation of \pjrp if and only if $f(N') \le -1$.
\end{proof}

Finally, we note the following relationship between priceability and \pjrp.\footnote{We note that any priceable committee of size exactly $k$ can always be made priceable with a budget $B > k$. Hence, \Cref{prop:priceable-rankpjrp} generalizes Proposition~1 by \citet{PeSk20a}.} 
\begin{proposition}\label{prop:priceable-rankpjrp}
 For instances with approval preferences, any priceable committee with a price system $\{B,p\}$ such that $B > k$ satisfies \pjrp.
\end{proposition}
\Cref{prop:priceable-rankpjrp} is a corollary of \Cref{prop:rank-price}, which we prove in \Cref{sec:ranked}.  
\Cref{prop:priceable-rankpjrp} implies that MES and EAR, as well as the following approval-based voting rules, satisfy PJR+: Phragmén's sequential rule \citep{BFJL16a}, the maximin support method \citep{SFFB22a}, and Phragmms \citep{CeSt21a}. Further, PAV also satisfies \pjrp \citep{AzLe21a}.

\begin{remark} \label{rem:pjrp-ejr}
Like priceability, \pjrp is incomparable to EJR: Phragmén's sequential rule satisfies \pjrp, but not EJR, and there exist committees which satisfy EJR, but not \pjrp \cite{AzLe21a}.   
\end{remark}

\subsection{Extended Justified Representation (EJR) Without Cohesiveness} 
 \label{sec:ejr+}

PJR is a rather weak axiom, and strengthening it to \pjrp does not make much of a difference, as we will see in \Cref{sec:exp-app}. Nevertheless, we can take the relationship between PJR and \pjrp 
as a blueprint for defining a robust strengthening of EJR. We call the resulting axiom \textit{\ejrp}. 

\begin{definition}[\ejrp] \label{def:ejrplus}
Given an instance with approval preferences, a feasible committee $W$ satisfies \emph{\ejrp} if there is no candidate $c \in C \setminus W$, group of voters $N'\subseteq N$, and $\ell \in \mathbb{N}$ with $\lvert N'\rvert \ge \ell \frac{n}{k}$ such that
\[
c \in \bigcap_{i \in N'} A_i \quad \text{and}\quad \lvert A_i \cap W \rvert  < \ell \,\text{ for all } i\in N'.
\]
\end{definition}

In words, for a violation of \ejrp it is sufficient to find a candidate who is approved by at least $\frac{\ell n}{k}$ voters, each of whom approves strictly less than $\ell$ candidates in the committee.\footnote{
In contrast to EJR violations, a violation of \ejrp does not necessarily correspond to a core deviation (see also \Cref{rem:core}).}
This candidate serves as a witness of the fact that the group $N'$ is not proportionally represented in the committee. 
\begin{figure}
    \centering
        \scalebox{0.8}{
    \begin{tikzpicture}
    [yscale=0.6,xscale=0.8,
    voter/.style={anchor=south}]
    
        \foreach \i in {1,...,8}
    		\node[voter] at (\i-0.5, -1) {$\i$};
        
        \draw[fill=orange!\clrstr] (0, 0) rectangle (4, 1);
        
        \draw[fill=violet!\clrstr] (0, 1) rectangle (4, 2);
        \draw[fill=teal!\clrstr] (0, 2) rectangle (7, 3);
        \draw[fill=magenta!\clrstr] (1, 3) rectangle (8, 4);
        \draw[fill=red!\clrstr] (3, 4) rectangle (8, 5);
        \draw[fill=blue!\clrstr] (5, 0) rectangle (8, 1);
        \draw[fill=purple!\clrstr] (6, 1) rectangle (8, 2);
        \node at ( 2, 0.5) {$c_{1}$};
        \node at ( 2, 1.5) {$c_{2}$};
        \node at ( 4, 2.5) {$c_{3}$};
        \node at ( 5, 3.5) {$c_{4}$};
        \node at ( 5.5, 4.5) {$c_{5}$};
        \node at ( 6.5, 0.5) {$c_{6}$};
        \node at ( 7, 1.5) {$c_{7}$};
    \end{tikzpicture}}
    \hspace{2em}
    \vrule
    \hspace{2em}
    \scalebox{0.8}{
    \begin{tikzpicture}
    [yscale=0.6,xscale=0.8, 
    voter/.style={anchor=south}]
    
         \foreach \i in {1,...,8}
    		\node[voter] at (\i-0.5, -1) {$\i$};
        
        \draw[fill=orange!\clrstr] (0, 0) rectangle (3, 1);
        \draw[fill=violet!\clrstr] (1, 1) rectangle (3, 2);
        \draw[fill=teal!\clrstr] (2, 2) rectangle (6, 3);
        \draw[fill=magenta!\clrstr] (2, 3) rectangle (6, 4);
        \draw[fill=red!\clrstr] (3, 0) rectangle (7, 1);
        \draw[fill=blue!\clrstr] (6, 1) rectangle (8, 2);
        \draw[fill=purple!\clrstr] (7, 0) rectangle (8, 1);
        \node at ( 1.5, 0.5) {$c_{1}$};
        \node at ( 2, 1.5) {$c_{2}$};
        \node at ( 4, 2.5) {$c_{3}$};
        \node at ( 4, 3.5) {$c_{4}$};
        \node at ( 5, 0.5) {$c_{5}$};
        \node at ( 7.5, 0.5) {$c_{7}$};
        \node at ( 7, 1.5) {$c_{6}$};
        
    \end{tikzpicture}}
    \caption{
    Approval profiles discussed in \Cref{ex:ejrp,ex:ejrp2}. Both instances have $8$ voters and candidates $c_1, \dots, c_7$. Voters correspond to integers and approve all candidates placed above them.
    }
    \label{fig:sec4instances}
\end{figure}
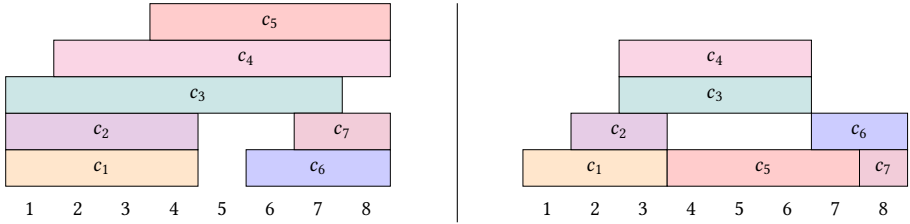

\begin{example}
To see how EJR+ is different from EJR, consider the two approval profiles illustrated in \Cref{fig:sec4instances}. Both profiles have $n=8$ voters. We consider $k = 4$, so that $\frac{n}k=2$. 

For the instance on the left, consider the committee $W = \{c_1, c_3, c_5, c_7\}$. This committee satisfies EJR, since all voters are covered, and every $2$-cohesive group contains at least one voter approving $2$ candidates in $W$. However, $W$ does not satisfy \ejrp. For this, consider the candidate $c_4$ and the group $\{2,3,5,6,7,8\} \subseteq N_{c_4}$. 
This group has $6=3\frac{n}k$ members and thus deserves to be represented by $3$ seats; however, no group member approves $3$ candidates in $W$. Hence, $c_4$ has a justified claim to be included. 

For the instance on the right, consider the committee $W' = \{c_1, c_2, c_3, c_7\}$. This committee satisfies EJR, since the only $2$-cohesive group $\{3,4,5,6\}$ is covered (voter $3$ even approves $3$ candidates in $W'$) and since the group $N' = \{4,5,6,7\}$ narrowly misses out on qualifying as a $2$-cohesive group. The group~$N'$ together with candidate $c_5$, however, witness a violation of \ejrp, since they together deserve two candidates, but each group member approves at most one candidate in $W'$.
\label{ex:ejrp}
\end{example}

By definition, \ejrp is a stronger requirement than EJR and \pjrp. 

\begin{proposition}
A committee $W$ satisfying \ejrp satisfies EJR and \pjrp as well.
\end{proposition}

Moreover, we can utilize existing proofs to show that both MES and PAV satisfy \ejrp.

\begin{proposition}
 Given an instance with approval preferences, PAV and MES satisfy \ejrp.
\end{proposition}
\begin{proof}
    Consider the proof that PAV satisfies EJR \citep{ABC+16a} and the proof that MES satisfies EJR \citep{PeSk20a}. Both proofs assume for contradiction that EJR is violated and argue that an EJR violation leads to the existence of a candidate $c$ who is approved by $\ell \frac{n}{k}$ voters, each with less than $\ell$ approved candidates in the committee, and then use this candidate to show a contradiction. This also follows from an \ejrp violation and, thus, both rules also satisfy \ejrp. 
\end{proof}

An important advantage of \ejrp over EJR is that the former can be verified in polynomial time, whereas the latter is coNP-complete to verify \citep{ABC+16a}. 

\begin{proposition}
 Given an instance with approval preferences and a committee $W$, it can be verified in polynomial time whether $W$ satisfies \ejrp.
\end{proposition}

\begin{proof}
It is sufficient to iterate over all $c \notin W$ and $\ell \in [k]$ and to check if there are at least $\ell \frac{n}{k}$ approvers of $c$ who approve strictly less than $\ell$ candidates in $W$. 
\end{proof}

The requirements on a \ejrp committee can easily be encoded as constraints of an Integer Linear Program (ILP), enabling computational experiments that optimize a linear function (e.g., the total approval score) over the space of all \ejrp committees. Similar experiments are not feasible for EJR and have previously only been possible for very weak axioms \citep{BFNK19a}. 

\begin{remark} \label{rem:core}
\ejrp is incomparable to alternative strengthenings of EJR known as \emph{fully justified representation (FJR)} \citep{PPS21a} and \emph{core stability} \citep{ABC+16a}. 
To see that \ejrp does not imply FJR or core stability, we note that MES satisfies \ejrp, but not FJR \citep{PPS21a}. 
To see that core stability does not imply \ejrp, consider an instance with $2$ voters and $3$ candidates $c_1, c_2, c_3$ such that the first voter approves $c_1$ and $c_2$ and the second voter approves $c_1$ and $c_3$. The committee $\{c_2, c_3\}$ is core-stable, but does not satisfy \ejrp. 
\end{remark}

\subsection{A Simple Greedy Procedure for Finding EJR+ Committees}
\label{sec:greedy}

\newcommand{\gjcr}{GJCR\xspace}

Yet another advantage of \ejrp is that its definition gives rise to a very simple greedy procedure that produces committees satisfying \ejrp (and thus EJR) in polynomial time. The \textit{Greedy Justified Candidate Rule (GJCR)}, formalized as \Cref{alg:gjcr}, greedily selects candidates witnessing a violation of \ejrp constraints, ordered by decreasing size of the corresponding voter group. 

\begin{example}\label{ex:ejrp2}
Consider again the two approval profiles illustrated in \Cref{fig:sec4instances} and let $k=4$. In the instance on the left, the rule would select both $c_3$ and $c_4$ for $\ell = 3$ and then terminate. In the instance on the right, the rule would select two candidates from $\{c_3, c_4, c_5\}$ for $\ell = 2$. If it selects $c_3$ and $c_4$, it will then take $c_1$ and $c_6$ for $\ell = 1$. Otherwise, it will just take $c_1$ for $\ell = 1$ and then terminate.
\end{example}

\begin{algorithm}[t]
\caption{Greedy Justified Candidate Rule}
\label{alg:gjcr}

 $W \gets \emptyset$\;
\For{$\ell$ in $k, \dots, 1$}{
\While{there is $c\notin W$: $\lvert \{i \in N_c \colon \lvert A_i \cap W \rvert < \ell\} \rvert \ge \frac{\ell n}{k}$ }{
 $W \gets W \cup \{c\}$\;
}
}
 return $W$\;
\end{algorithm}

\begin{proposition} \label{prop:GJCR}
    The Greedy Justified Candidate Rule selects a committee of size at most $k$ and satisfies \ejrp.
\end{proposition}

\begin{proof}
To see that the committee constructed by the rule satisfies \ejrp, note that any candidate witnessing a violation of \ejrp for a group of size $\frac{\ell n}{k}$ would be selected by the rule in iteration $\ell$. 

To show that the rule selects at most $k$ candidates, we show how to construct a price system for any committee in the output of the rule, in which the budget is at most $k$ and in which the price of a candidate is $1$. From the existence of this price system, it follows that at most $k$ candidates could have been selected. We start by assigning each voter a budget of $\frac{k}{n}$. To construct the payments, consider any step in which a candidate $c$ is selected. At this step, let $N'$ denote the set $\{i \in N_c \colon \lvert A_i \cap W \rvert < \ell\}$ of voters of size $|N'| \ge \frac{\ell n}{k}$ that is responsible for the addition of $c$. We divide the price of $c$ equally among $N'$ and reduce the budget of each voter by $\frac{1}{\lvert N'\rvert}$. To show that this does not overspend the budget, we notice that in previous iterations, only sets of size at least $\frac{\ell n}{k}$ were considered. Since voters in $N'$ approve at most $\ell-1$ previously selected candidates, each voter in $N'$ has paid at most $(\ell-1) \frac{k}{\ell n}$ before. After adding the payment for $c$, the total amount of budget that was spent by the voter is at most $(\ell-1) \frac{k}{\ell n} + \frac{k}{\ell n} = \frac{k}{n}$. Thus, the budget of voters is never overspent, and at most $k$ candidates can be bought. 
\end{proof}

Our greedy rule allows us to have a simple and fast baseline algorithm for \ejrp (and EJR). This algorithm is arguably easier to understand (and to implement) than any other known algorithm satisfying EJR. Such a simple rule is also useful in the analysis of \ejrp and its properties, and in the design of new algorithms: As we show in \Cref{sec:query}, the greedy rule can be used to find faster query procedures in the setting of \citet{HKP+23a}. 

Furthermore, the simplicity of the Greedy Justified Candidate Rule allows us to show that \ejrp (and thus EJR) satisfy the following property, which can be considered a weak form of committee monotonicity.\footnote{A rule satisfies \textit{committee monotonicity} (a.k.a. \textit{house monotonicity} or \textit{committee enlargement monotonicity}) if selected candidates for a given committee size $k$ are also selected for committee size $k+1$. It is an open question whether EJR is compatible with committee monotonicity \citep[page~105]{LaSk22a}.}  
It was previously unknown that this property holds for EJR. 

\begin{restatable}{proposition}{propcommon} \label{prop:com-mon}
    Consider an approval profile over a set $C$ of candidates. For all $k < m$, there exist a committee~$W$ of size $|W|=k$ and candidate $c \in C\setminus W$ such that 
    $W$ satisfies \ejrp w.r.t. committee size $k$ and $W \cup \{c\}$ satisfies \ejrp w.r.t. committee size $k+1$.
\end{restatable}
\begin{proof}
    Let $W'$ be the output of \gjcr for committee size $k+1$. If $\lvert W' \rvert < k+1$, we are already done, as we can simply fill up the committee arbitrarily. If, on the other hand, $\lvert W'\rvert = k+1$, let~$c$ be the last candidate selected by \gjcr and let $N'$ and $\ell$ be the respective voter set and $\ell$-value. 
    Consider the price system for $W'$ that is constructed by \gjcr (as described in the proof of \Cref{prop:GJCR}). Since we select exactly $k+1$ candidates, we know that every voter must spend their whole budget, since the budget is exactly enough for $k+1$ candidates. Since every voter paid at most $\frac{k+1}{\ell n}$ per candidate and paid exactly $\frac{k+1}{n}$ in total, the voters outside $N'$ must approve at least $\ell$ candidates in $W' \setminus \{c\}$, while the voters in $N'$ must approve exactly $\ell - 1$ in $W' \setminus \{c\}$.
    
    Now assume that $W \coloneqq W' \setminus \{c\} $ does not satisfy \ejrp for committee size $|W|=k$. Since every voter approves at least $\ell - 1$ candidates in $W$, there must be an $\ell' \ge \ell, c' \notin W$ and $N'' \subseteq N_{c'}$ with $\lvert N'' \rvert \ge \frac{\ell' n}{k}$ such that $\lvert A_i \cap W \rvert < \ell'$ for all $i \in N''$. If $c = c'$, we get that $\ell' = \ell$ and thus $N'' \subseteq N'$, since only voters in $N'$ can approve less than $\ell$ candidates in $W$. 
    If $c \neq c'$, we know that $c'$ was not chosen by \gjcr. Thus, in this case $\ell' = \ell$ and $N'' \subseteq N'$ must hold as well. It follows that $\lvert N' \rvert \ge \frac{\ell n}{k} > \frac{\ell n }{k+1}$. Therefore, the voters in $N'$ could not have spent their whole budget and $W'$ cannot be of size $k+1$, a contradiction. 
\end{proof}
\subsection{Experiments with Approval Ballots}
\label{sec:exp-app}

In order to assess how demanding the notions studied in Sections~\ref{sec:pjr+} and~\ref{sec:ejr+} are, we ran experiments with randomly generated instances. As a measure for the discriminative power of an axiom, we consider the likelihood that a randomly drawn committee satisfies the axiom; stronger axioms yield lower values.

\paragraph{Setup}
Following the recent work of \citet{SFJ+22a}, we sample profiles from two different statistical cultures: the resampling model and the truncated urn model. We provide a brief description of these models below, and refer to the paper by \citet{SFJ+22a} for details. 
\begin{itemize}
    \item 
In the \textit{resampling model} with parameters $p$ and $\phi$, an approval ballot $A$ approving $\lfloor p m \rfloor$ candidates is generated uniformly at random. Following this, for each voter $i$ and candidate~$c$, the approval of~$i$ to~$c$ is kept as in $A$ with probability $1 - \phi$ and resampled with probability $\phi$, with the probability of approval being $p$. This model generalizes the impartial culture. 
\item
 The \textit{disjoint model} has three parameters, $p, \phi$, and $g$. The model partitions the candidates into~$g$ groups $C_1, \dots, C_g$. Then, for each voter it samples $j \in [g]$ uniformly at random and runs the same procedure as in the resampling model with parameters $(p, \phi)$ and the approval ballot $C_j$ as the start vote. 
\item
In the \textit{noise model} with parameters $\phi$ and $p$, we sample one vote approving $\lfloor p m \rfloor$ uniformly at random. For each voter, a new vote is then generated proportional to $\phi^d$, where~$d$ is the distance of the new vote to the original vote.
\item
In the \textit{truncated urn model} with parameters $\alpha$ and $p$, we sample according to the Pólya-Eggenberger Urn Model with parameter $\alpha$, and truncate the corresponding ranking to approve the top $\lfloor p m \rfloor$ candidates.
\end{itemize}

For each model, we consider four different values for $p$: $0.2$, $0.4$, $0.6$, and $0.8$. For the resampling, noise, and disjoint models, we let $\phi$ vary between $0.01$ and $1$ (in steps of $0.01$) and for the truncated urn model, we let $\alpha$ vary between $0.01$ and $1$ (in steps of $0.01$). For each combination of parameter values, we sample $400$ instances with $100$ voters and $50$ candidates. For each instance, we sample one committee of size $k=10$ uniformly at random and check whether it satisfies PJR, EJR, \pjrp, and \ejrp.

\paragraph{Results}
The results of our experiments are presented in \Cref{fig:app4,fig:app3,fig:app1,fig:app2}.
We observe that the established axioms PJR and EJR are quite easy to satisfy for a lot of parameters. This is similar to the findings of \citet{BFNK19a} and \citet{SFJ+22a}. Moreover, the difference between PJR+ and PJR is often very small. 
\ejrp, on the other hand, is almost always much more discriminating than all the other axioms. This is particularly true when the instances are ``far away'' from impartial culture (which corresponds to parameter value $\phi = 1$ in the resampling model). Interestingly, the incomparability of \pjrp and EJR (see \Cref{rem:pjrp-ejr}) is also visible in some of the graphs.

\newcommand{\figsize}{0.38}
\begin{figure*}[b]
    \centering
    \includegraphics[scale = \figsize]{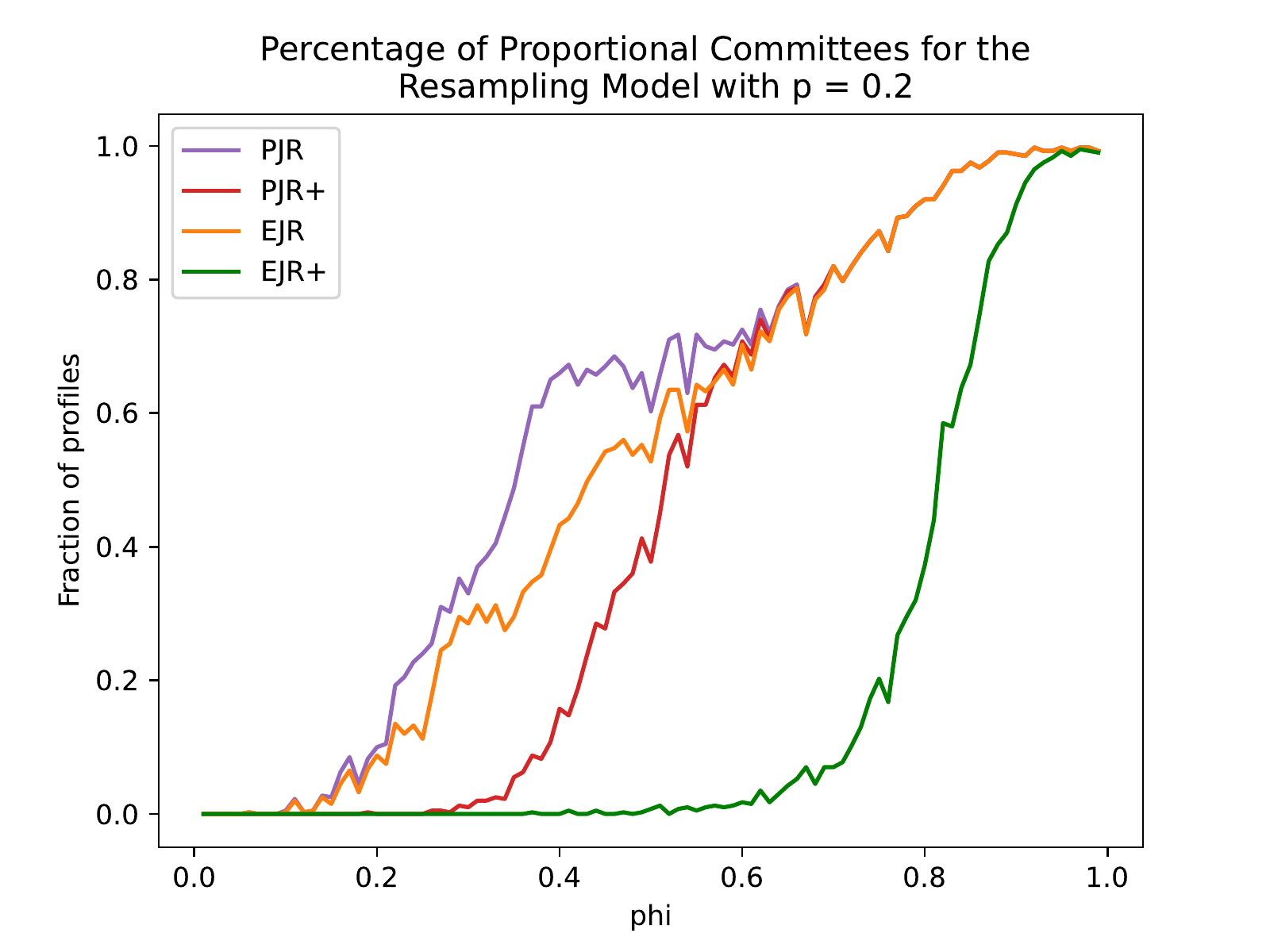}
    \includegraphics[scale = \figsize]{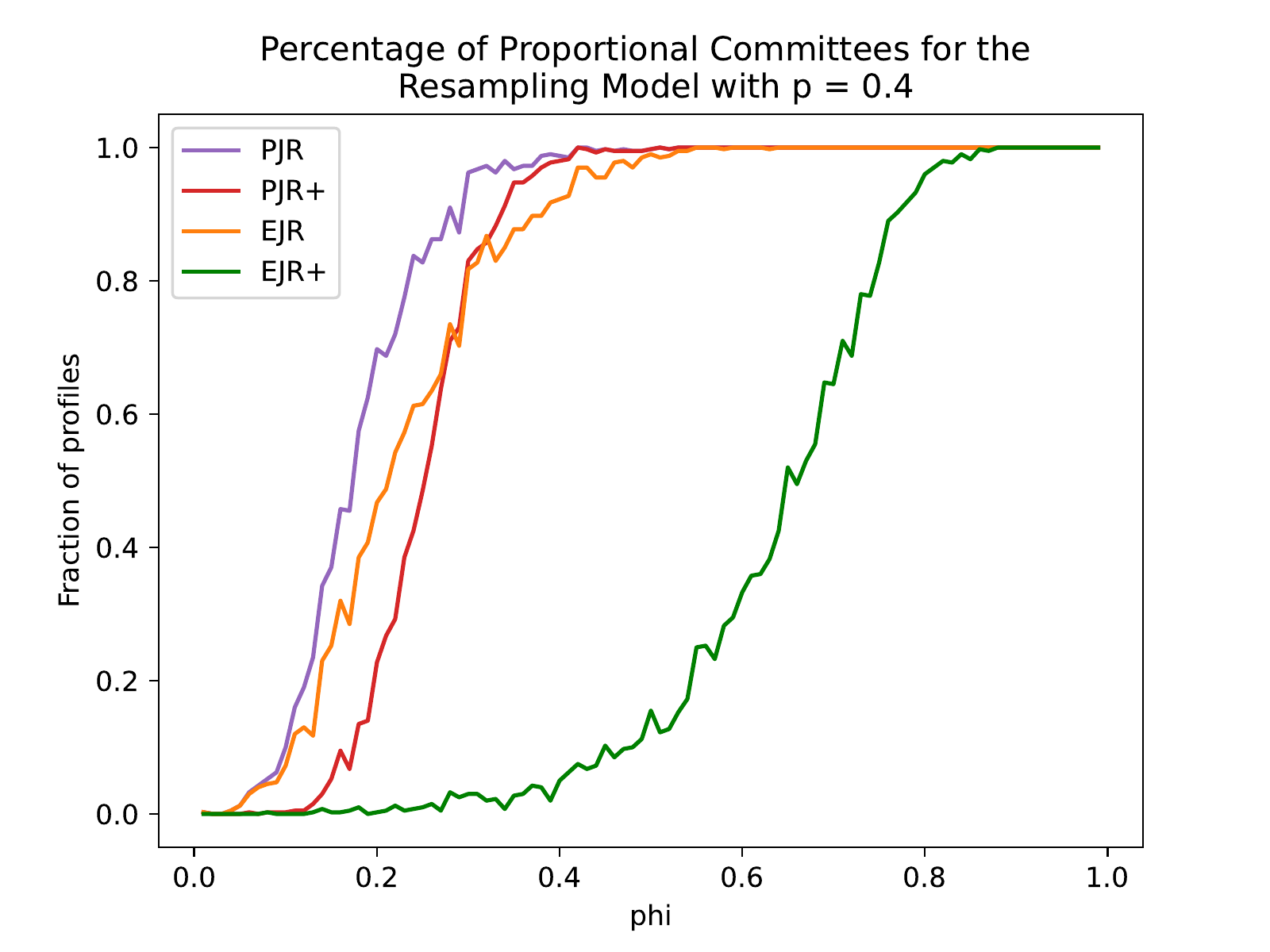}
    \includegraphics[scale = \figsize]{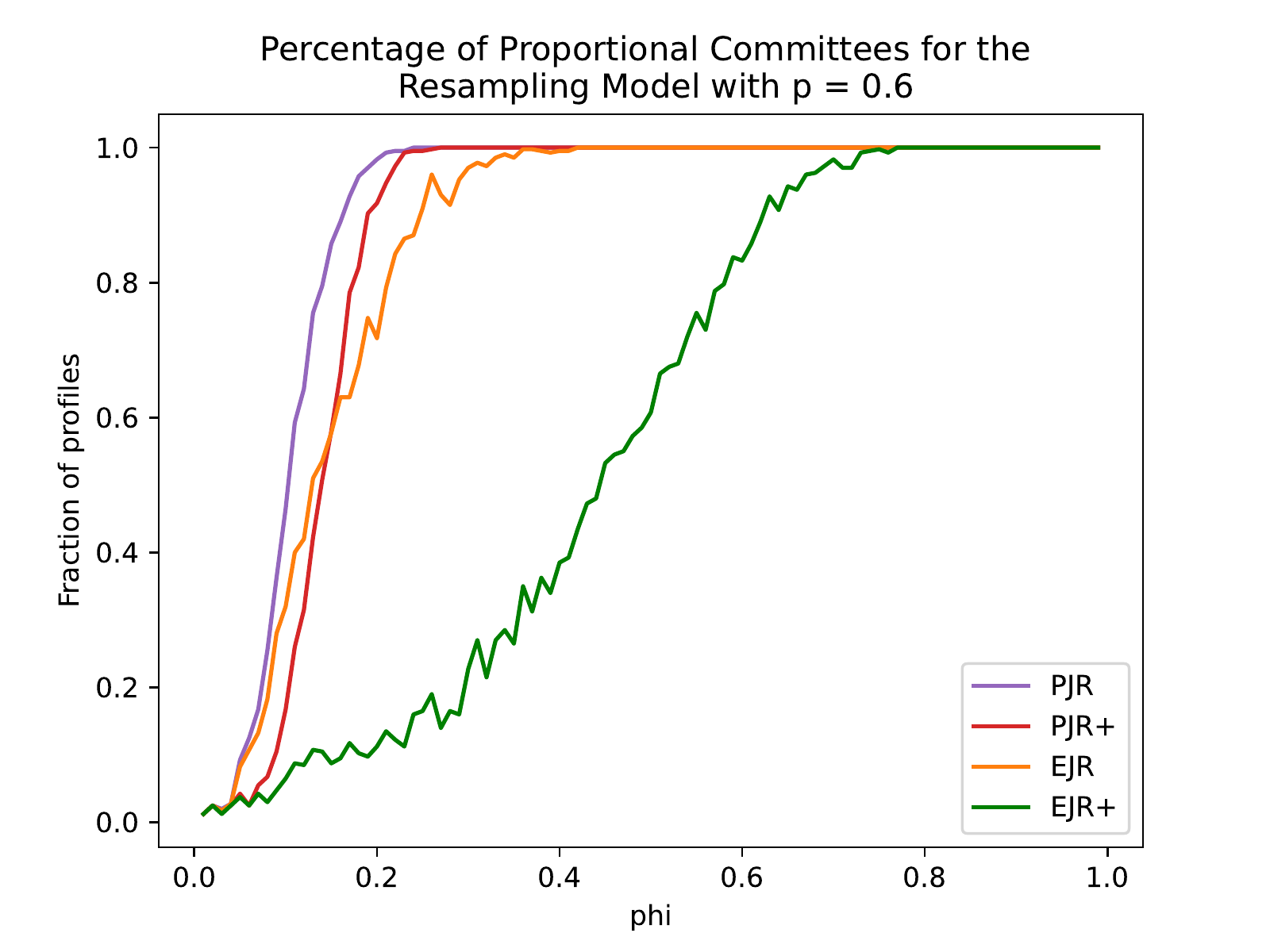}
    \includegraphics[scale = \figsize]{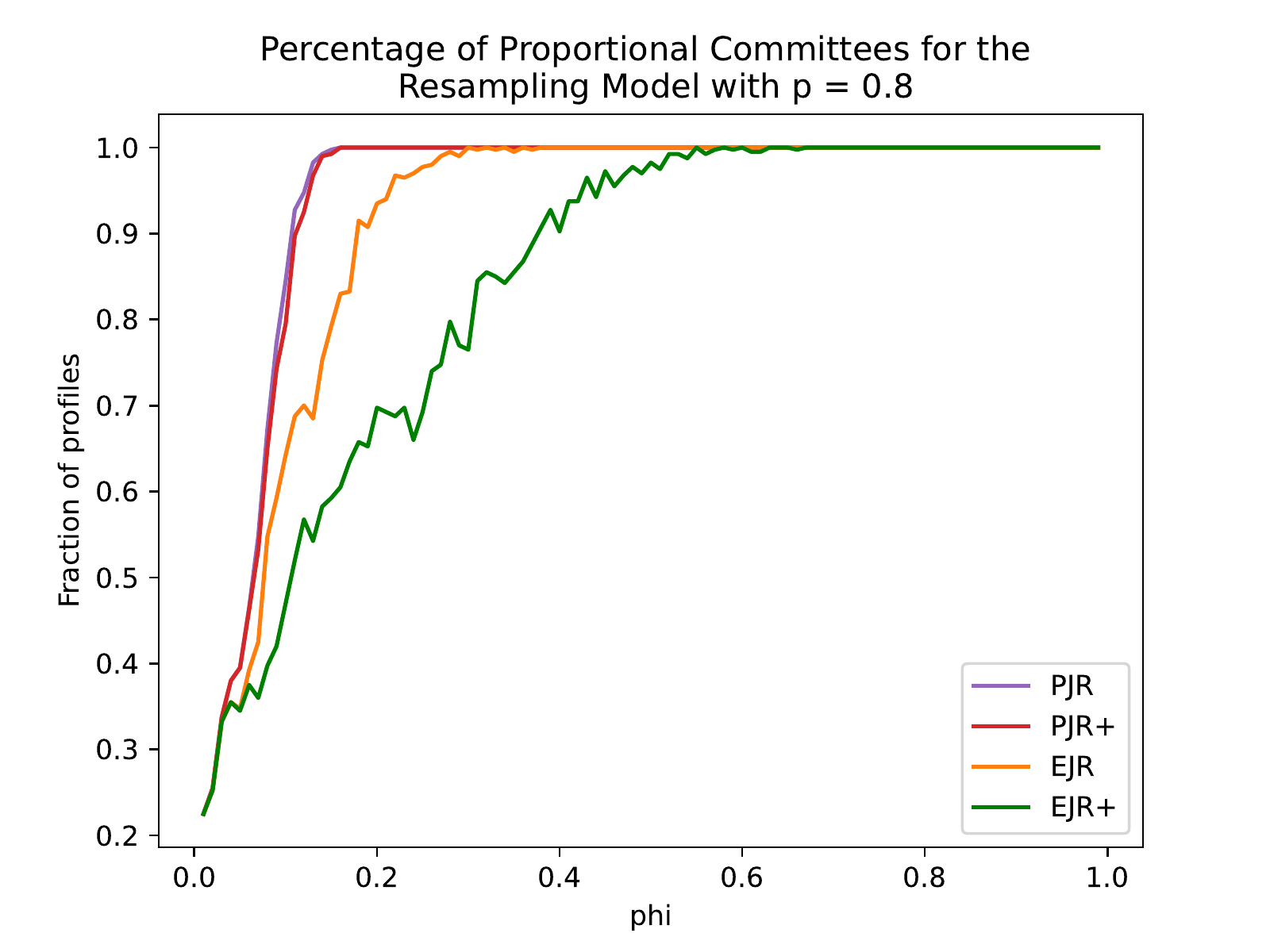}
        \vspace{-0.2cm}
    \caption{Experimental results for the resampling model}
    \label{fig:app2}
\end{figure*}
\begin{figure*}
    \centering
    \includegraphics[scale = \figsize]{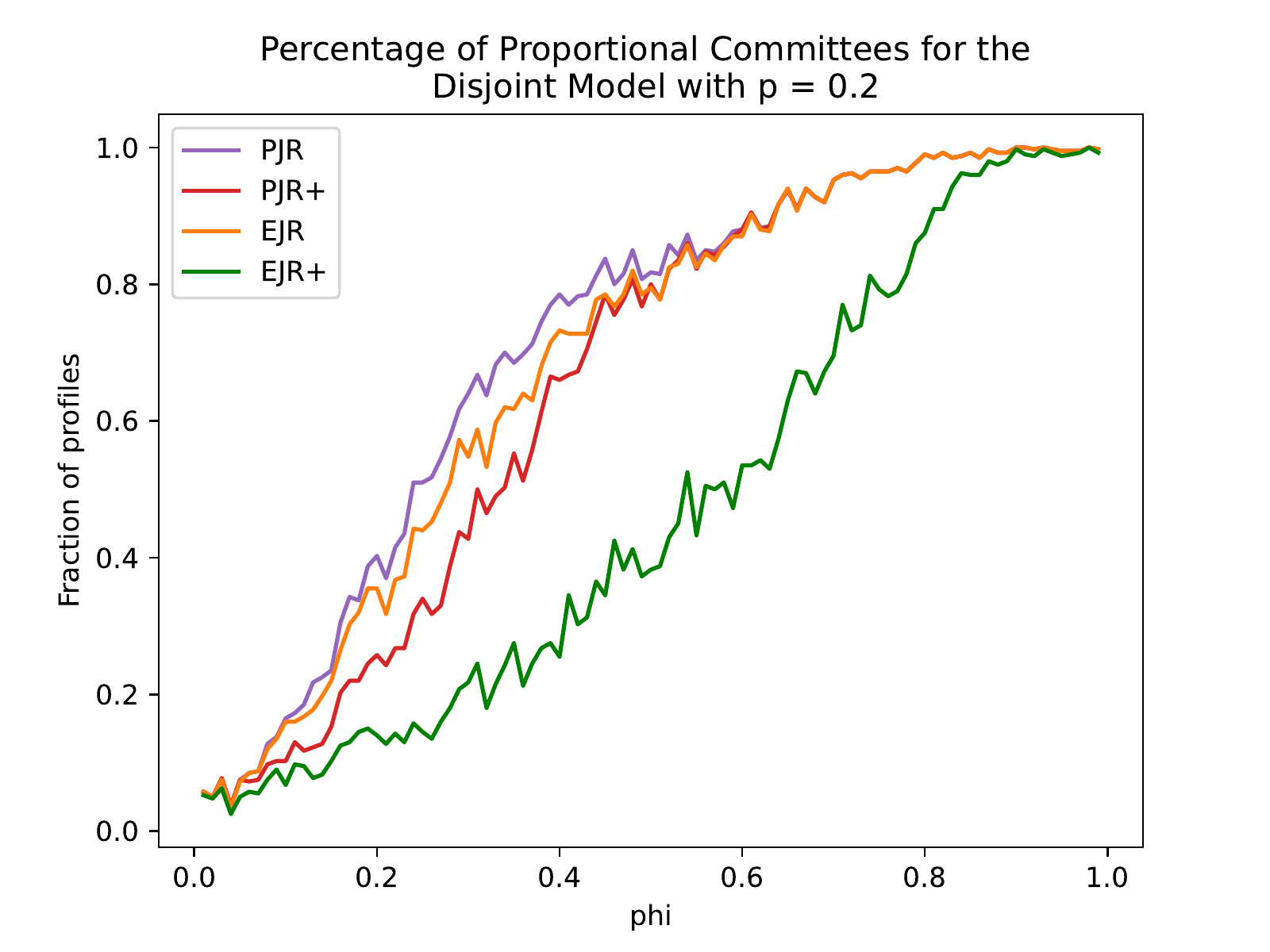}
    \includegraphics[scale = \figsize]{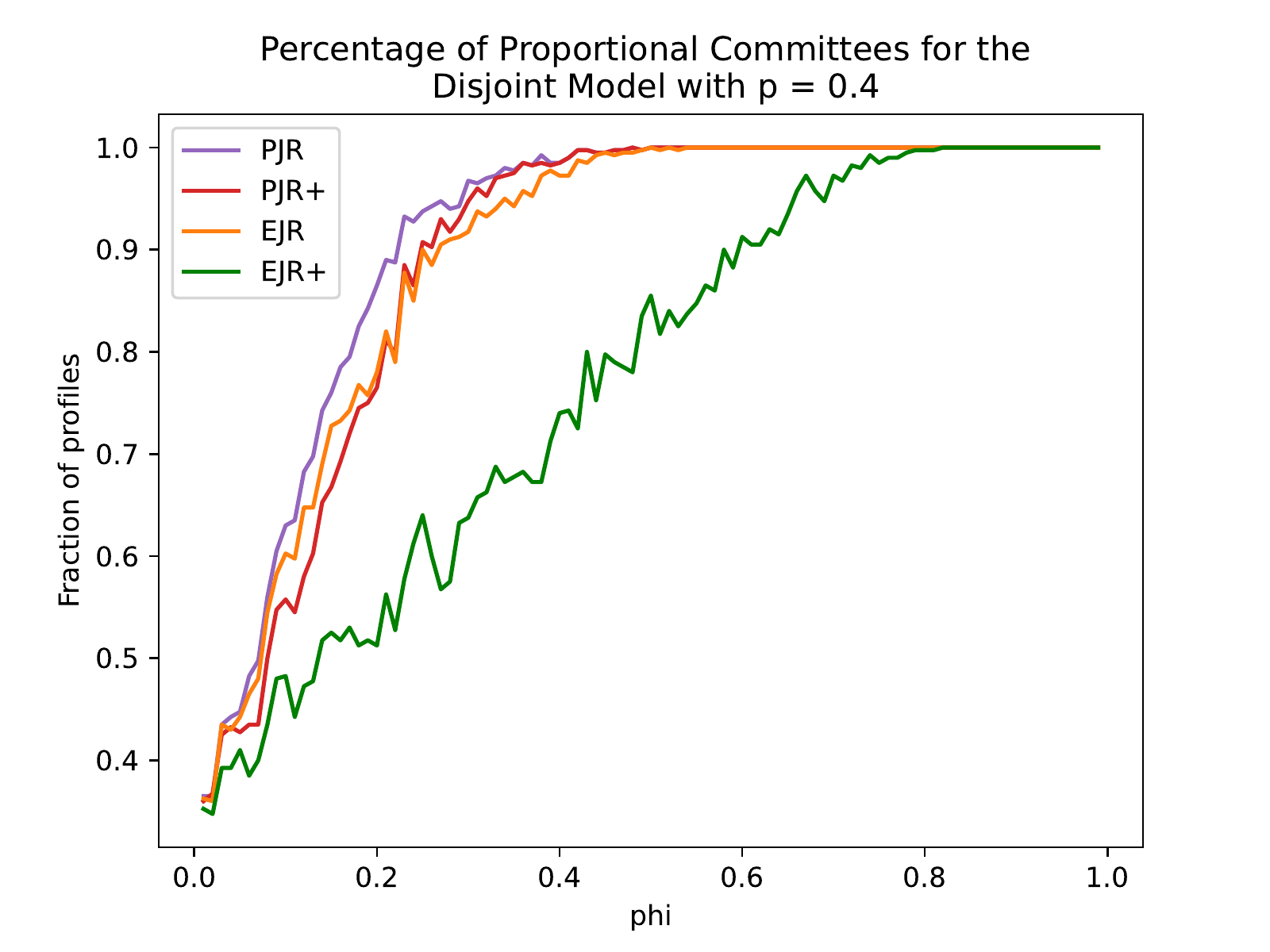}
    \includegraphics[scale = \figsize]{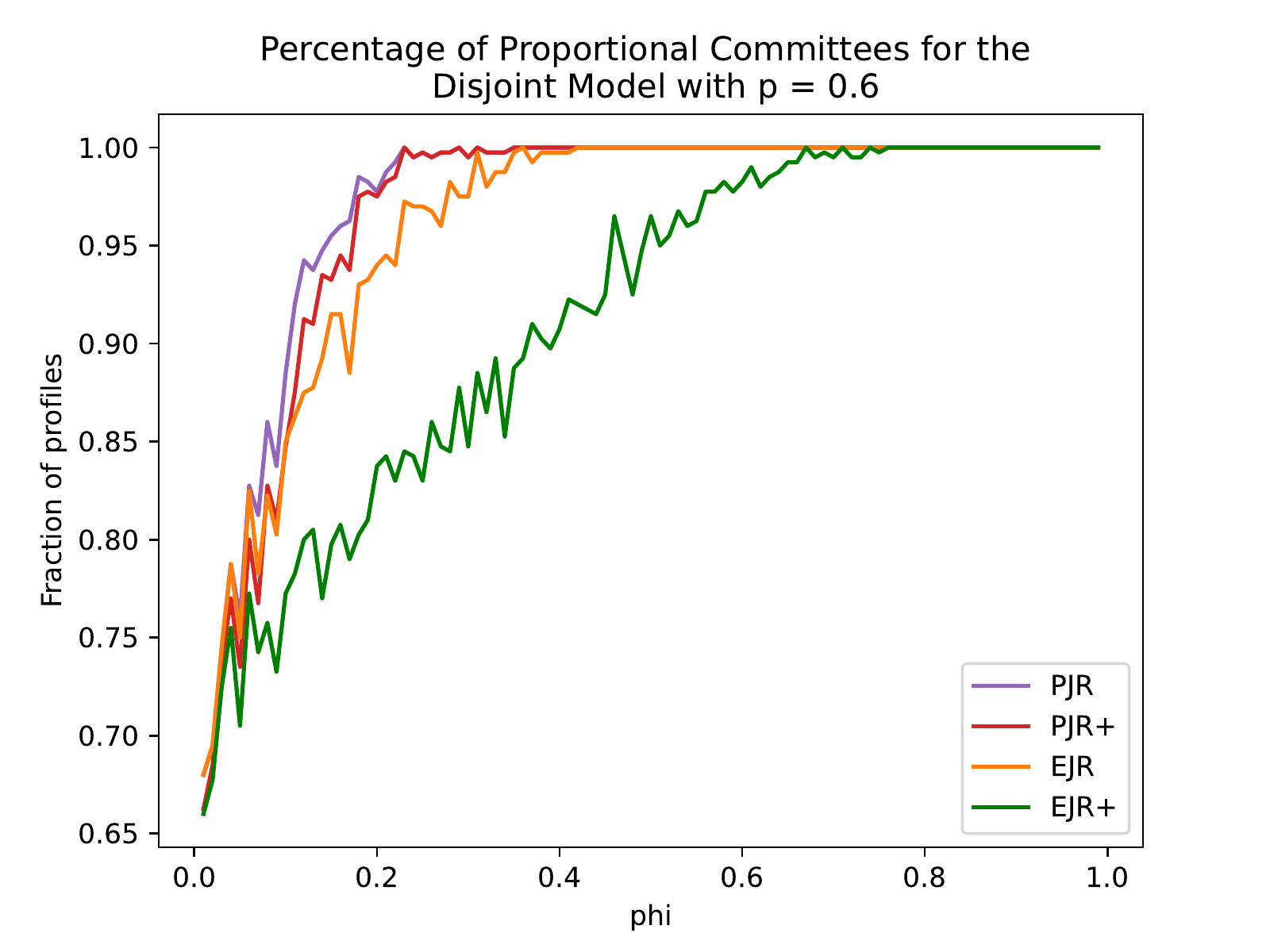}
    \includegraphics[scale = \figsize]{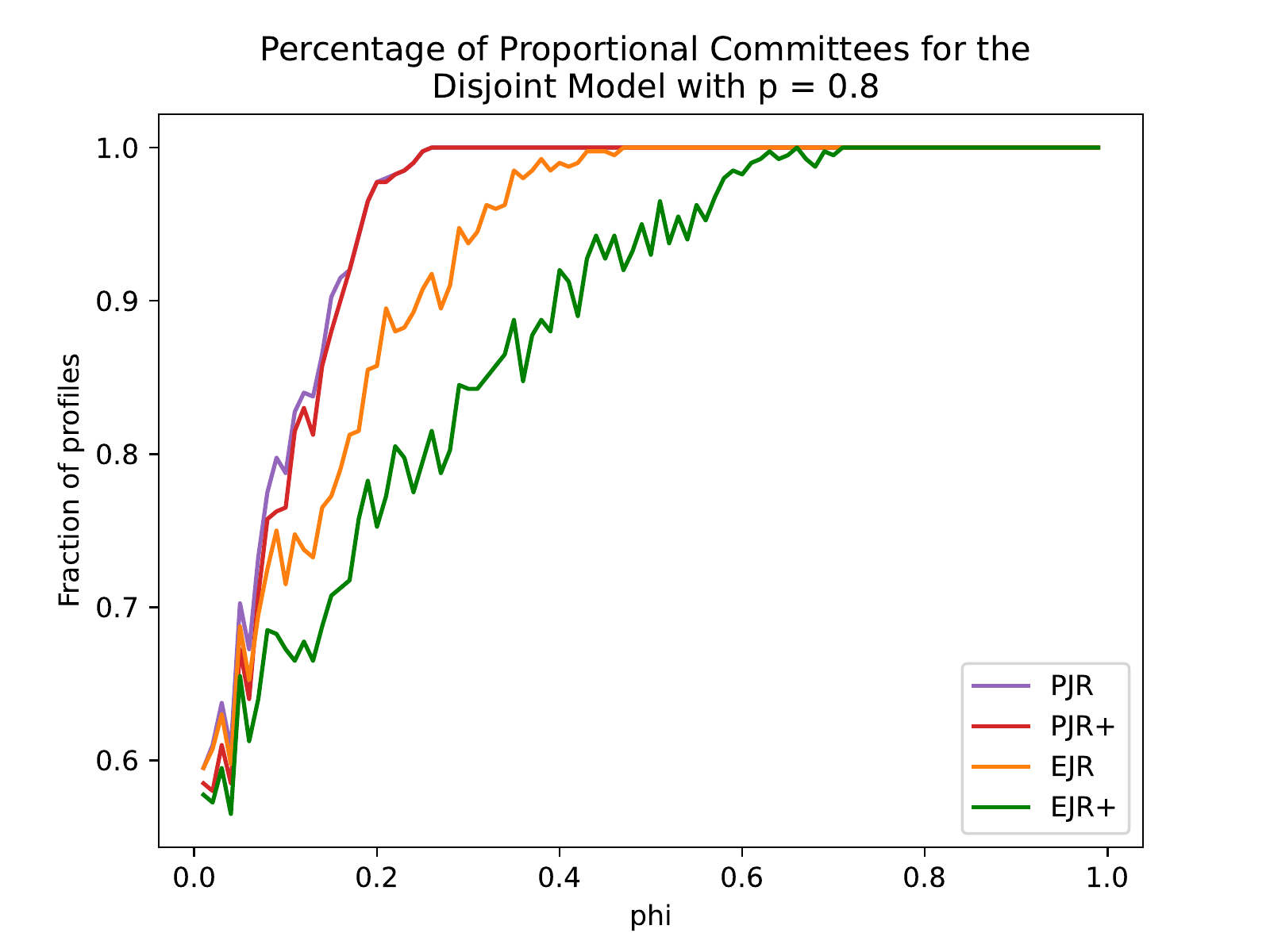}
    \vspace{-0.2cm}
    \caption{Experimental results for the disjoint model}
    \label{fig:app4}
\end{figure*}
\begin{figure*}
    \centering
    \includegraphics[scale = \figsize]{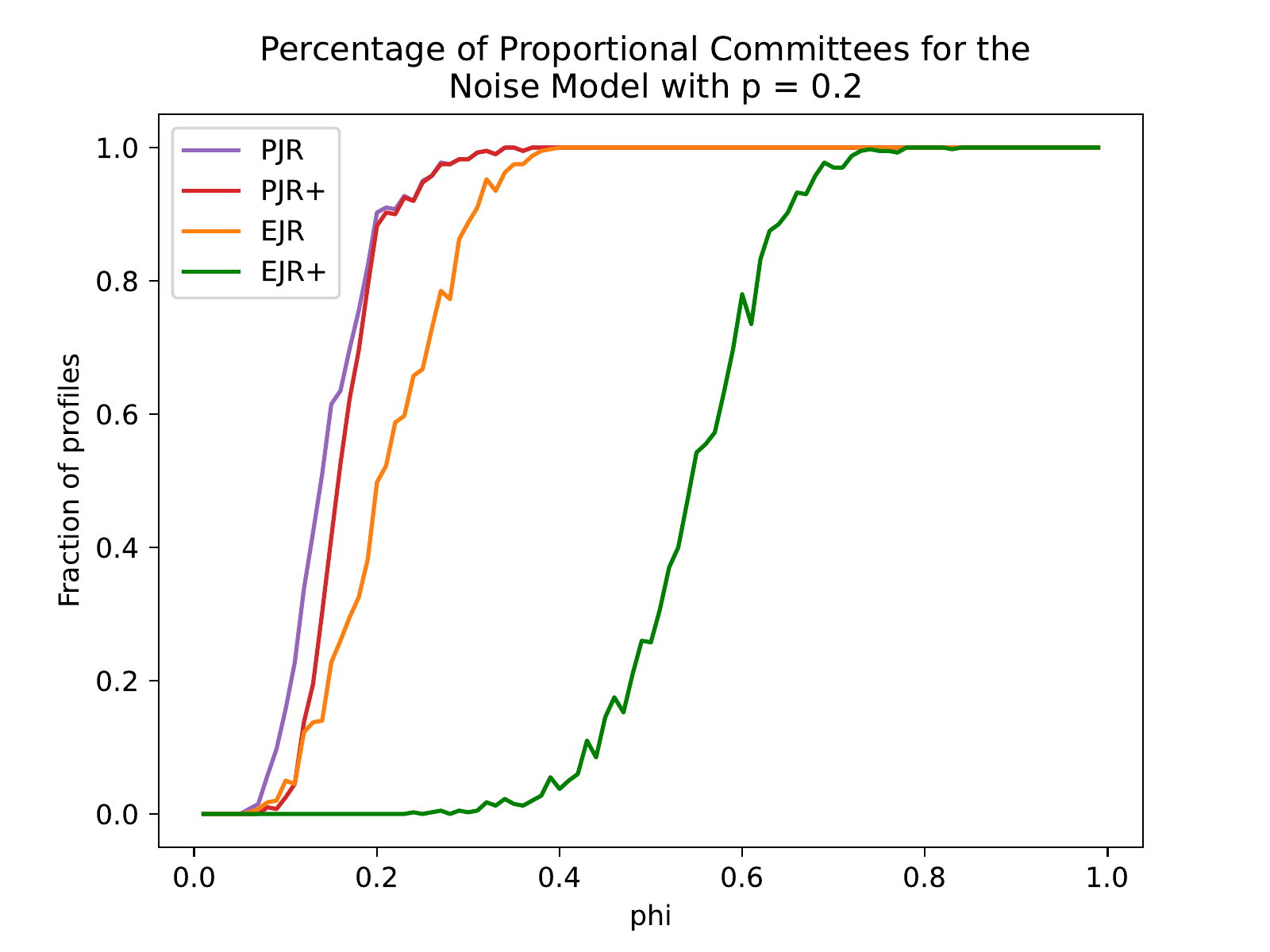}
    \includegraphics[scale = \figsize]{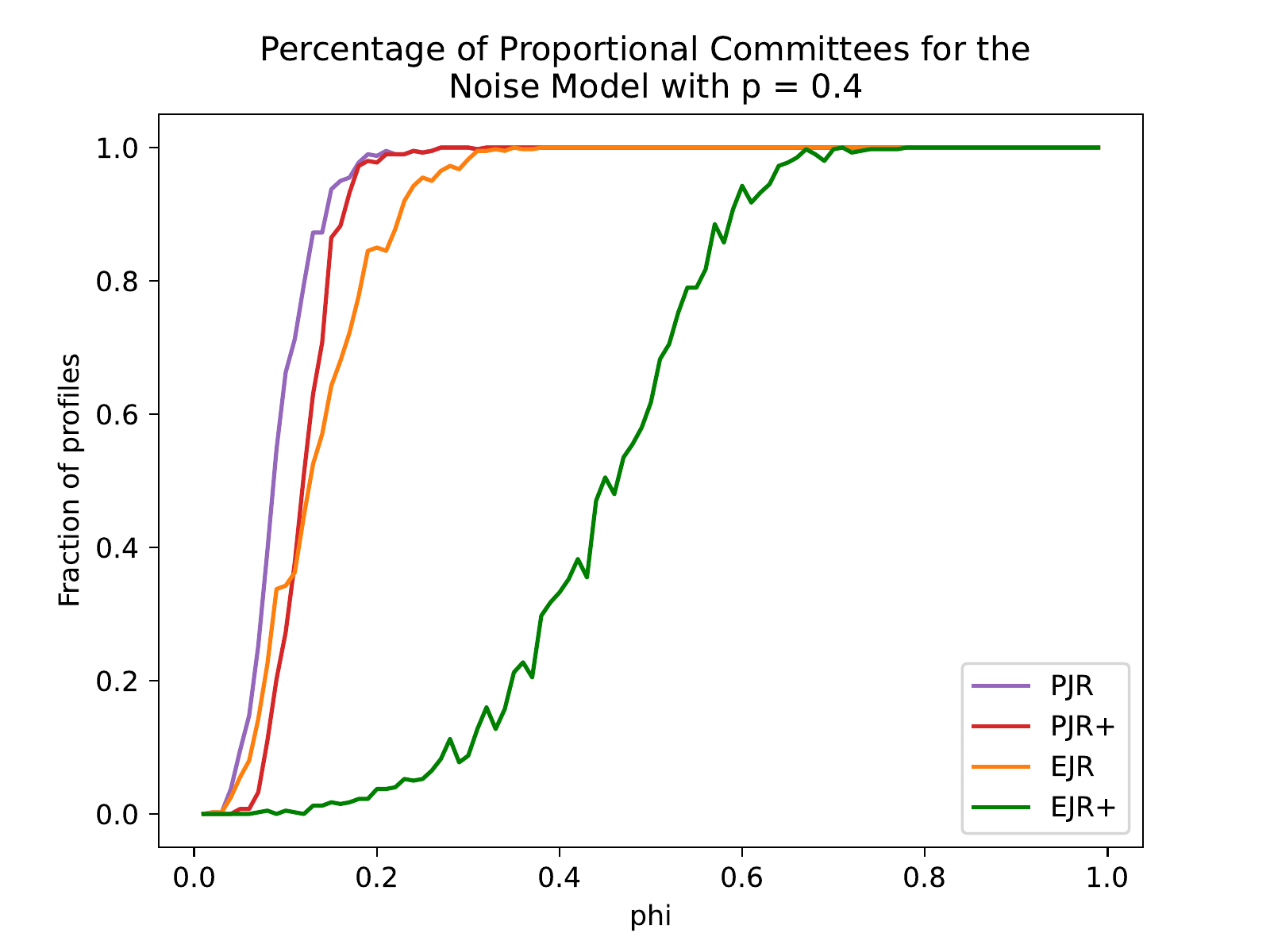}
    \includegraphics[scale = \figsize]{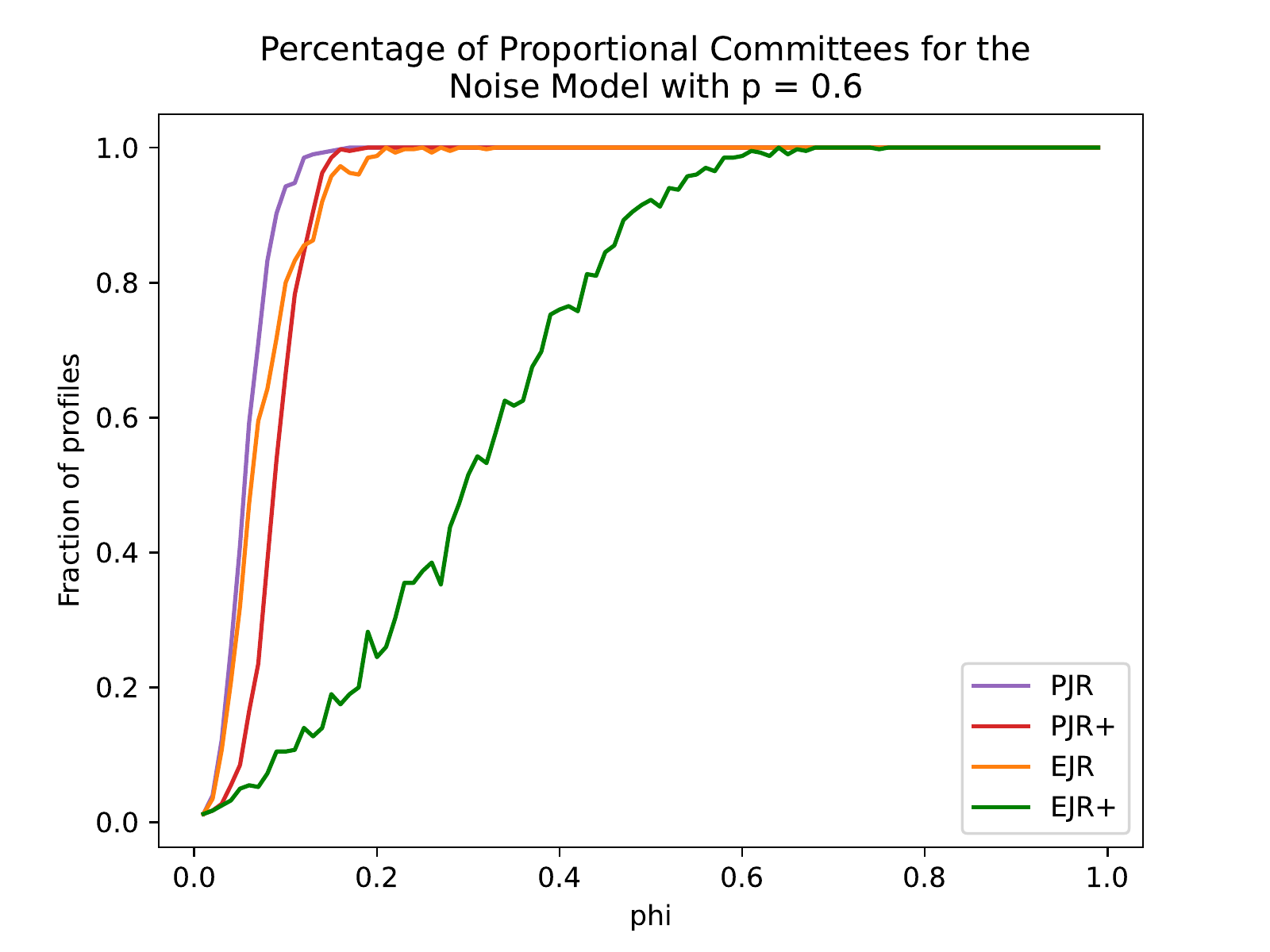}
    \includegraphics[scale = \figsize]{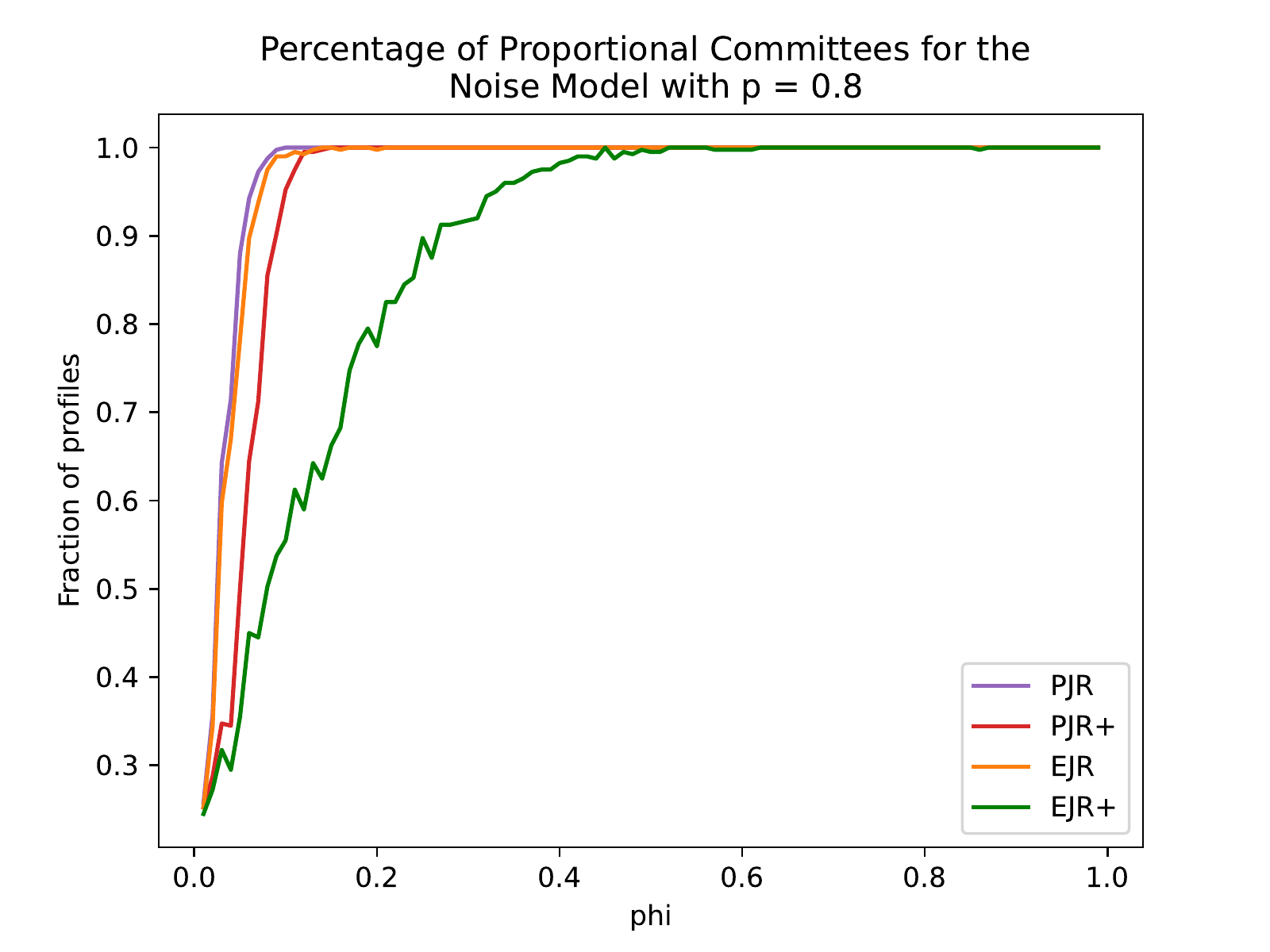}
    \vspace{-0.2cm}
    \caption{Experimental results for the noise model}
    \label{fig:app1}
\end{figure*}
\begin{figure*}
    \centering
    \includegraphics[scale = \figsize]{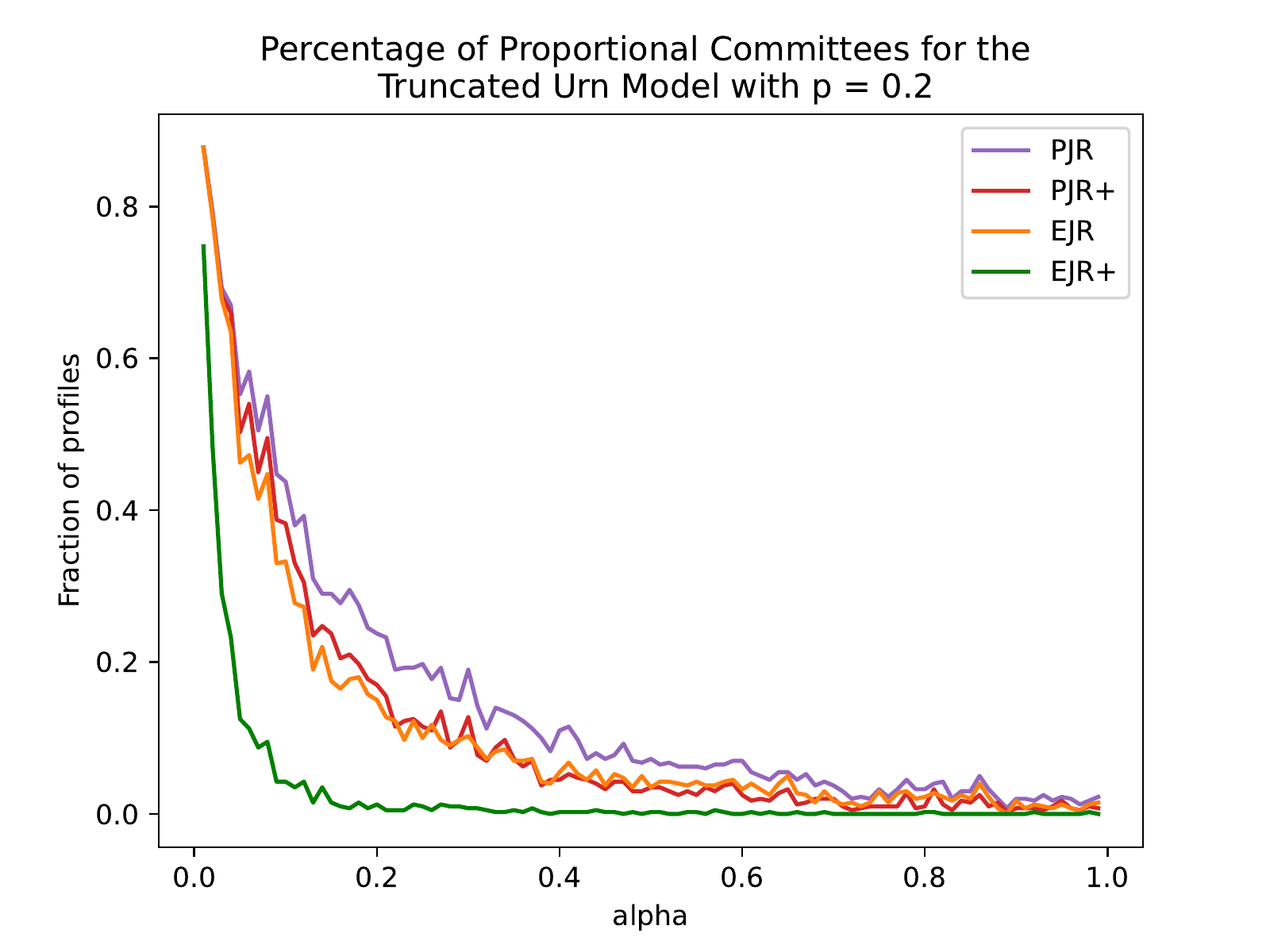}
    \includegraphics[scale = \figsize]{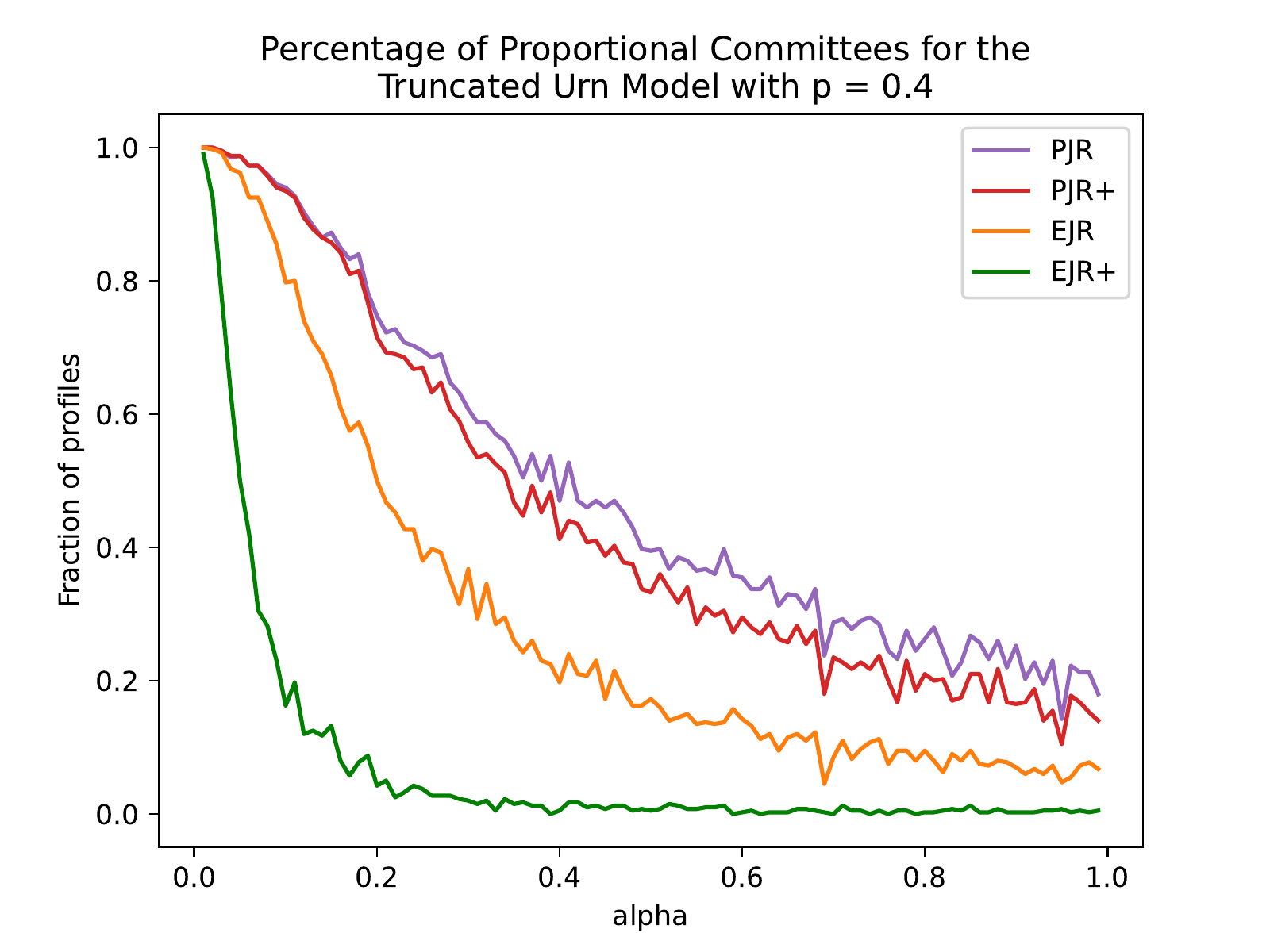}
    \includegraphics[scale = \figsize]{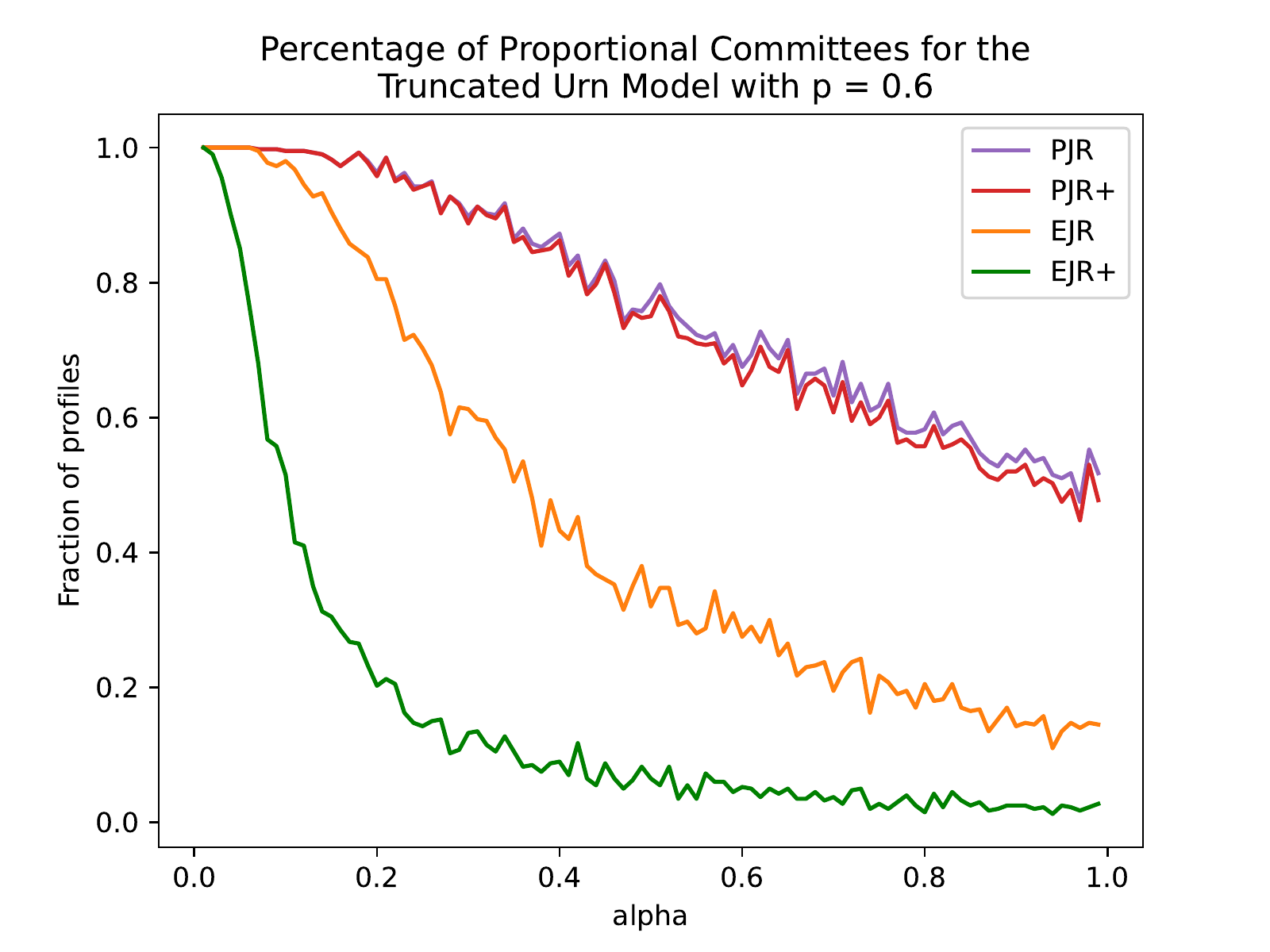}
    \includegraphics[scale = \figsize]{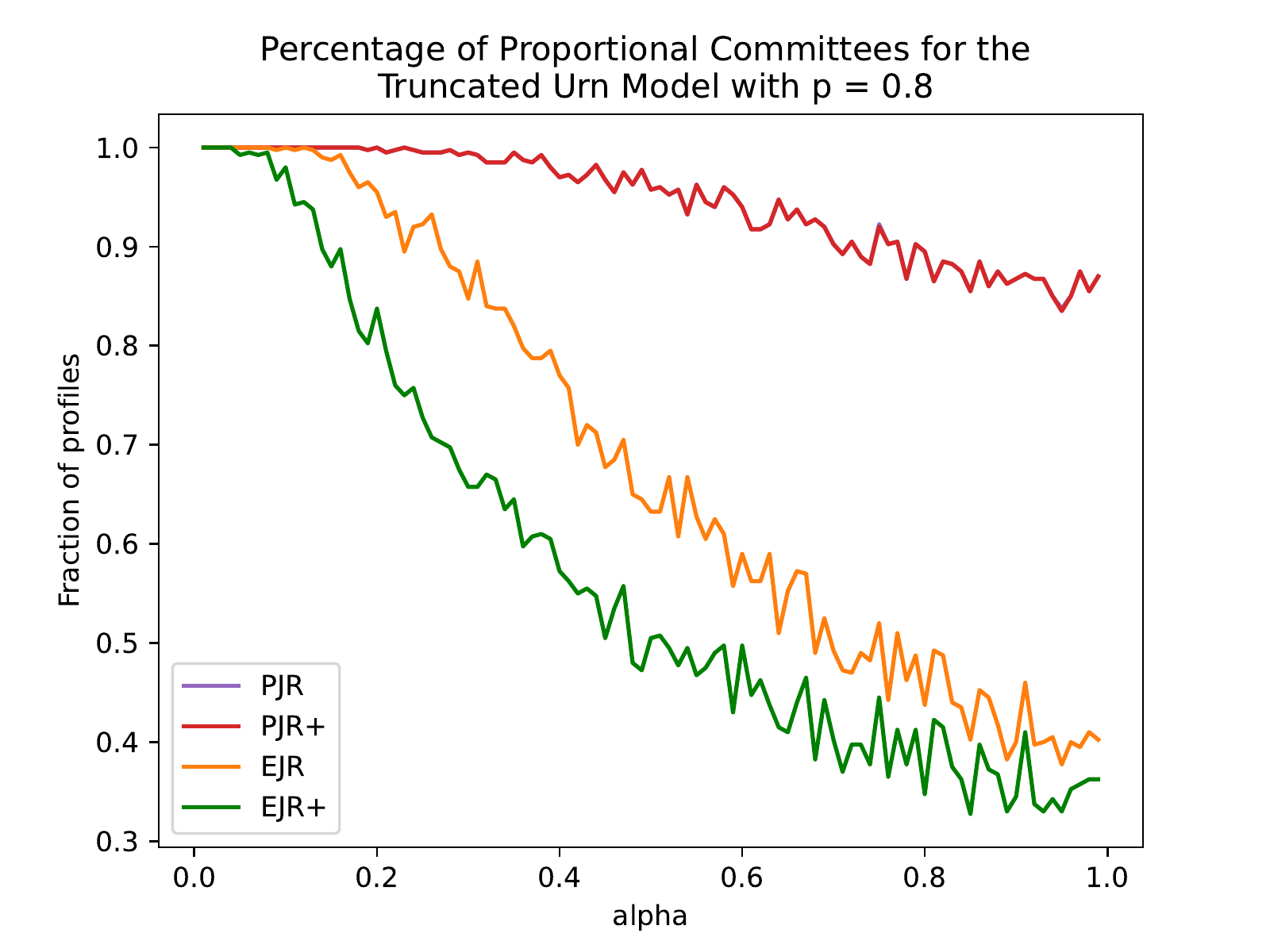}
    \vspace{-0.2cm}
    \caption{Experimental results for the truncated urn model}
    \label{fig:app3}
\end{figure*}

\clearpage

\section{Ranked Ballots}
\label{sec:ranked}

In this section, we consider multiwinner voting based on ranked preferences. Our results hold for general weak preferences, and we often mention the important special case of strict preferences. 
To motivate the need for more robust proportionality axioms, consider the following example with strict preferences.\footnote{A similar (slightly more complex) example was given by \citet{AzLe20a}. We discuss their example in Appendix~\ref{app:AzLe}.}

\begin{example}
Consider an instance with four voters, $m=6$, $k = 2$, and the following strict preferences:
\begin{align*}
   1&:c_1 \succ c_2 \succ c_3 \succ c_4 \succ c_5\succ c_6 \\ 
   2&:c_5 \succ c_2 \succ c_3 \succ c_4 \succ c_6 \succ c_1 \\ 
   3&:c_4 \succ c_3 \succ c_2 \succ c_6 \succ c_1 \succ c_5\\ 
   4&:c_6 \succ c_3 \succ c_2 \succ c_5 \succ c_1 \succ c_4
\end{align*}
Since there are no non-trivial (generalized) solid coalitions, even the strongest axiom IPSC (which coincides with PSC for strict preferences; see \Cref{obs:ipsc_psc}) enforces nothing. Thus, every committee satisfies IPSC, even $\{c_4,c_6\}$. This committee could also be selected by STV (by first eliminating $c_2$ and $c_3$, then eliminating $c_5$, selecting $c_4$, eliminating $c_1$, and selecting $c_6$). 
However, this committee seems gravely unfair to the first two voters, 
who can make the claim that they should be represented by half of the committee. Even though they do not form a solid coalition, there exists a candidate, $c_2$, that represents those voters much better than any committee member. Hence, requiring voters to form (generalized) solid coalitions is too demanding a requirement in this instance. \label{exmp:ipsc}
\end{example}

A first attempt to define a notion of proportionality without reference to solid coalitions was made by \citet{AEFLS17} in the setting with strict preferences. According to their definition, a committee $W$ is \textit{locally stable} if, for every group $N'$ of voters with $\lvert N' \rvert \ge \frac{n}{k}$, there is no single unchosen candidate $c \notin W$ with $\{c\} \succ_i W$ for all $i \in N'$. 
However, this notion is not always satisfiable, and it is computationally intractable to decide whether an instance admits a locally stable committee \citep{AEFLS17}.

Recall from \Cref{sec:prelims-weak} that generalized PSC and IPSC can be defined using the upper contour set $\overline{C'}(N')$ of a set $C'$ of candidates w.r.t. a set $N'$ of voters. An intuitive idea to generalize these axioms could consist in requiring that
\[
\left\lvert \overline{\{c\}}(N') \cap W \right\rvert  \ge \left \lfloor \frac{\lvert N'\rvert k}{n} \right \rfloor
\]
holds for any $N' \subseteq N$ and unchosen candidate $c \in \bigcap_{i \in N'} A_i$. This is, however, more demanding than local stability, and committees satisfying this notion may fail to exist. 

In order to slightly relax the requirement above, we return to the concept of ranks.
(Recall that only acceptable candidates are assigned finite ranks.)
Using the concept of ranks, we can translate any instance with weak preferences into a collection of instances with approval preferences in a straightforward way: For each possible rank $r$, the approval set of a voter $i$ with preference relation $\succeq_i$ can be defined as the set of all candidates that have been assigned a rank of at most $r$ by voter $i$. To make this precise, consider an instance $(N, C, P, k)$ with weak preferences. For each $r \in [m]$, 
define the approval profile $P^r=(A_1^r, \dots, A_n^r)$ via $A^r_i = \{c \in C \colon \rank(i,c)\le r\}$. This translation results in $m$ instances\footnote{The approval profiles $P^1$, ..., $P^{m}$ have a nested structure and are not necessarily distinct. In particular, if the profile $P$ has dichotomous preferences to begin with, then $P^r$ is identical to $P$ for all $r \in [m]$.}  
with approval preferences, given by $(N, C, P^1, k), \dots, (N, C, P^{m}, k)$. Intuitively, an outcome which is proportional for the instance with weak preferences should also be proportional for all the approval instances derived from it. The following definition formalizes this idea.

\begin{definition}[Rank-\pjrp] \label{def:rankPJR+}
Given an instance $(N, C, P, k)$ with weak preferences, a feasible committee $W$ satisfies \rpjrp  if, for every $r \in [m]$, $W$ satisfies \pjrp in the instance $(N, C, P^r,k)$. 
\end{definition}

This property can also be formulated without reference to PJR+.

\begin{observation}              
    A feasible committee $W$ satisfies \rpjrp if and only if there is no candidate $c \in C\setminus W$, rank $r \in [m]$, and group of voters $N'$ with $\lvert N'\rvert \ge \frac{\ell n}{k}$ such that $\rank(i,c) \le r$ for all $i \in N'$ and  $\lvert\{c \in W \colon \rank(i,c) \le r \text{ for some } i\in N'\}\rvert < \ell$.
\end{observation}

\begin{example} \label{ex:rank-pjrp}
    Consider again the instance given in \Cref{exmp:ipsc}. 
    The committee $\{c_4,c_6\}$ does not satisfy \rpjrp, since candidate~$b$ is justified to be included in the committee via voter group $\{1,2\}$ in the approval instance $(N,C,P^2,k)$, where each voter approves their two top-ranked candidates. An analogous violation of \rpjrp can be found for candidate $c_3$.  
    A committee satisfying \rpjrp has to include at least one of $c_2$ or $c_3$. More precisely, \rpjrp requires that either (i) $c_2$ and one of $\{c_3,c_4,c_6\}$ or (ii) $c_3$ and one of $\{c_1,c_2,c_5\}$ is selected.
\end{example}

For instances with approval preferences, rank-\pjrp is equivalent to \pjrp. 
To show that \rpjrp is always satisfiable, we turn to the concept of priceability \citep{PeSk20a} and define a generalization to weak preferences. 

\subsection{Priceability for Ranked Ballots}

The goal of this section is to generalize the notion of priceability (see Section ~\ref{sec:prelim-approval}) from approval instances to instances with general weak preferences.  
First, we note that the definition of a price system (i.e., the constraints \textbf{C1} to \textbf{C5} in \Cref{def:priceability}) 
generalize immediately from the approval setting to the ranked setting. 
Axiom \textbf{C5} however, which is required to relate price systems to proportionality axioms, does not give us any proportionality guarantees in the ranked setting: Consider an instance with two voters, both of whom rank $c_1$ higher than $c_2$, and let $k = 1$.  Selecting $c_2$, a clearly suboptimal choice, results in a priceable committee, even with budget $B > 1$. Unsurprisingly, the committee $\{c_2\}$ violates all our previously defined proportionality axioms for ranked preferences.

In order to relate priceability to proportionality, we thus need to strengthen \textbf{C5}. A first idea to do that would be to require that
no candidate can be bought using unspent money and money spent on candidates ranked worse than them. Formally, this corresponds to requiring  
    \vspace{-0.5ex}
    \begin{align*}
    \underbrace{\sum_{i \in N \colon c \in A_i} \left(\frac{B}{n} - \sum_{c' \in C}p_i(c')\right)}_{\text{unspent money}} \quad + \quad \underbrace{\sum_{i \in N}\left(\sum_{c' \in C \colon c\succ c' } p_i(c')\right)}_{\text{money spent on worse candidates}} \le 1 \quad \text{ for all $c \notin W$.} 
    \end{align*} 
 While seemingly reasonable, this axiom is not always satisfiable, as a committee satisfying it would also be locally stable. 
 Therefore, we relax the above inequality by considering ranks. Let $N_c^r = \{i \in N \colon \rank(i,c) \le r\}$ be the set of voters ranking candidate $c$ at rank $r$ or higher. The following constraint, to which we refer as $\textbf{C5}_{\textbf{rank}}$, requires that the unspent money of $N_c^r$ plus the money this group spent on candidates ranked worse than $r$ is not enough to buy~$c$.
\begin{align*}
    \textbf{(C5}_{\textbf{rank}}\textbf{)} \quad \underbrace{\sum_{i \in N_c^r} \left(\frac{B}{n} - \sum_{c' \in C}p_i(c')\right)}_{\text{unspent money}} \, + \, \underbrace{\sum_{i \in N_c^r}\left(\sum_{c' \in C \colon c' \notin A_i^r} p_i(c')\right)}_{\text{\parbox{3.5cm}{\centering money spent on candidates\\ ranked worse than $r$}}} \le 1 \quad \text{ for all } c \notin W \text{ and } r \in [m].
\end{align*}

We call a committee $W$ \emph{rank-priceable} if there is a price system for $W$ satisfying $\textbf{C5}_{\textbf{rank}}$.\footnote{Since $\textbf{C5}_{\textbf{rank}}$ implies \textbf{C5}, there is no need to require \textbf{C5} explicitly.}

\begin{definition}[Rank-Priceability] \label{def:priceability}
Given an instance with weak preferences, a committee $W$ is \emph{rank-priceable}  
if there exist a $B>0$ and functions $p_i\colon C \rightarrow [0, \frac{B}{n}]$ satisfying \emph{\textbf{C1}--\textbf{C4}} and \emph{$\textbf{C5}_{\textbf{rank}}$}.
\end{definition}

It is easy to see that these constraints can be verified using a linear program. 

\begin{proposition}
 Given an instance with weak preferences and a committee $W$, it can be verified in polynomial time whether $W$ is rank-priceable. 
\end{proposition}

Next, we show that rank-priceability implies \pjrp.
\begin{restatable}{proposition}{proprprice} \label{prop:rank-price}
 Any rank-priceable feasible committee with a price system $\{B,p\}$ such that $B > k$ satisfies \rpjrp.
\end{restatable}
\begin{proof}
Assume that $W$ does not satisfy \rpjrp. Then there is a rank $r \in [m]$ such that $W$ does not satisfy \pjrp in the approval instance $(N, C, P^r,k)$. Thus, there is a group of voters $N'$ and a candidate $c \in \bigcap_{i \in N'}A_i^r\setminus W$ with 
\[
\left\lvert \bigcup_{i \in N'} A^r_i \cap W\right\rvert < \left \lfloor \frac{\lvert N'\rvert k}{n} \right \rfloor.
\] 
We get that 
\begin{align*}
\sum_{i \in N_c^r} &\left(\frac{B}{n} - \sum_{c' \in C}p_i(c')\right)+ \sum_{i \in N_c^r}\left(\sum_{c' \in C \colon c' \notin A_i^r} p_i(c')\right) \\
    &\ge\frac{\lvert N'\rvert B}{n} + \sum_{i \in N'} \sum_{c' \in C} - p_i(c') + \sum_{i \in N_c^r}\left(\sum_{c' \in C \colon c' \notin A_i^r} p_i(c')\right) \\
    &= \frac{\lvert N'\rvert B}{n} - \sum_{i \in N'}\left(\sum_{c' \in A_i^r } p_i(c')\right) 
\,>\, \frac{\lvert N'\rvert k}{n} - \left\lvert \bigcup_{i \in N'} A^r_i \cap W\right\rvert \,\ge\, 1 \text,
\end{align*}
which is a contradiction to $\textbf{C5}_{\textbf{rank}}$. 
\end{proof}

In order to prove that rank-\pjrp is always satisfiable, it is therefore sufficient to identify rules that produce rank-priceable committees. We show that EAR is such a rule. 

\begin{restatable}{proposition}{earprice}
 Any committee in the output of EAR is rank-priceable for some $B > k$.
\end{restatable}
\begin{proof}
First, we observe that during the execution of EAR (\Cref{alg:ear}) we indeed construct a pricesystem satisfying \textbf{C1}--\textbf{C4} for budget $B = k$. Further, if there was an unpicked candidate at rank $r$ with a budget of $1$ or more left among the voters ranking it at least as good as $r$, this candidate would be selected. Hence, the inequality in $\textbf{C5}_{\textbf{rank}}$ is strict. It follows that there exists $\varepsilon > 0$ such that the inequality also holds for $B = k + \varepsilon$. Thus, each EAR committee is indeed rank-priceable.
\end{proof}

As a consequence, EAR (and thus also MES) satisfies \rpjrp. 
We can further show that rank-priceability is a strictly stronger requirement than \rpjrp, even for strict preferences.
\begin{restatable}{proposition}{pricestronger}
There exist committees $W$ which satisfy \rpjrp, but not rank-priceability.
\end{restatable}
\begin{proof} 
    Consider the following instance with $n=2$ voters, $m=6$ candidates, and $k = 3$.
    \begin{align*}
        1&:c_1 \succ c_2 \succ c_4 \succ c_5 \succ c_3 \\
        2&:c_1 \succ c_3 \succ c_4 \succ c_5 \succ c_2
    \end{align*}
    The committee $\{c_1,c_4,c_5\}$ satisfies \rpjrp. Candidates $c_2$ and $c_3$ cannot witness a violation of \rpjrp in the second rank, since both voters only deserve $\frac{3}{2}$ candidates. However, independent of the price system, there is one voter spending at least $1$ on $c_4$ and $c_5$ and thus witnesses a rank-priceability violation for rank $2$.
\end{proof}

In our experiments (see \Cref{sec:exp-ranked}), committees satisfying rank-\pjrp, but not rank-priceability were extremely rare. In other words, rank-priceability is not a significant strengthening of rank-\pjrp and mainly serves the technical purpose of establishing the existence of rank-\pjrp committees. 

Finally, we note that \rpjrp is indeed a stronger notion than IPSC, even for strict preferences.

\begin{restatable}{proposition}{rpjrpstronger}
 Any committee satisfying \rpjrp also satisfies IPSC, but there are committees which satisfy IPSC, but not \rpjrp. The latter holds even for strict preferences.
\end{restatable}
\begin{proof}
As stated earlier, the latter part follows from \Cref{exmp:ipsc,ex:rank-pjrp}. 
To see that \rpjrp implies IPSC, let $W$ be a committee satisfying \rpjrp and $N' \subseteq N$ be a set of voters forming a generalized solid coalition over a set $C' \subseteq C$ of candidates. 
Then, for $r = \lvert C'\rvert$, we see that $ \bigcup_{i \in N'} A_i^r$ is precisely $\overline{C'}(N')$, since any candidate with at most $r-1$ candidates ranked better, must be ranked equal to a candidate from $C'$. Further, by definition it holds that $C' \subseteq A_i^r$ for all $i \in N'$ and thus, there must be a candidate $c \in \bigcap_{i \in N'} A_i^r \setminus W$. Hence, due to \rpjrp, at least $\frac{\lvert N'\rvert k}{n}$ candidates must be selected from $ \bigcup_{i \in N'} A_i^r$ and thus from $\overline{C'}(N')$. 
\end{proof}

To the best of our knowledge, \rpjrp is the first proportionality axiom for strict preferences that is always satisfiable and violated by STV (see \Cref{exmp:ipsc,ex:rank-pjrp}). Therefore, our results can be interpreted as an axiomatic argument against STV (and in favor of rules such as EAR and MES). 

\subsection{Other Rank-Based Notions}

Analogously to rank-\pjrp, one can define rank-versions for other approval-based axioms. In light of the positive results from \Cref{sec:approval}, it is particularly tempting to consider rank-\ejrp.  

\begin{definition}[Rank-\ejrp]
Given an instance $(N, C, P, k)$ with weak preferences, a feasible committee~$W$ satisfies rank-\ejrp if, for every $r \in [m]$, $W$ satisfies \ejrp in the instance $(N, C, P^r,k)$.
\end{definition}

However, there exist instances where rank-\ejrp (and even rank-EJR) is unsatisfiable.

\begin{example}
Consider an instance with $n=2$, $m=4$, $k = 2$, and preferences $c_1 \succ_1 c_2 \succ_1 c_3 \succ_1 c_4$ and $c_4 \succ_2 c_2 \succ_2 c_3 \succ_2 c_1$. 
For $r = 1$, rank-\ejrp would require both $c_1$ and $c_4$ to be selected, while for $r = 3$, either $c_2$ or $c_3$ needs to be selected.
\end{example}

One could also consider $\rank$-PJR, for which we note the following relationships.

\begin{restatable}{proposition} {pjrrelations}
For weak preferences, \rpjrp implies $\rank$-PJR. Furthermore, $\rank$-PJR is incomparable to IPSC and implies generalized PSC. 
\end{restatable}
\begin{proof} 
    First, we notice that \rpjrp implies $\rank$-PJR by definition, since \pjrp implies PJR. To see that IPSC does not imply $\rank$-PJR, we refer to \Cref{exmp:ipsc}. Here, the committee $\{d,f\}$ satisfies IPSC but not $\rank$-PJR. 
    Similarly, PJR is strictly weaker than IPSC for approval preferences and hence, also does not imply it for weak-rankings. 

    To see that $\rank$-PJR implies generalized PSC, let $N'$ with $\lvert N' \rvert \ge \frac{\ell n}{k}$ be a group of voters forming a solid coalition over $C'$ with $\lvert \overline{C'}(N') \cap W \rvert < \min(\ell, \lvert C' \rvert)$. Then, in the instance $(N,C,P^{|C'|},k)$ the group $N'$ is $\ell$-cohesive, but at most $\ell - 1$ candidates are selected out of the union of their approval sets. Thus, PJR is violated in $(N,C,P^{|C'|},k)$ and hence also $\rank$-PJR. Thus, $\rank$-PJR implies generalized PSC.
\end{proof}

For strict preferences, $\rank$-PJR implies PSC and IPSC, but is not implied by them.

\renewcommand{\figsize}{0.42}

\begin{figure}[b]
    \centering
\includegraphics[scale = \figsize]{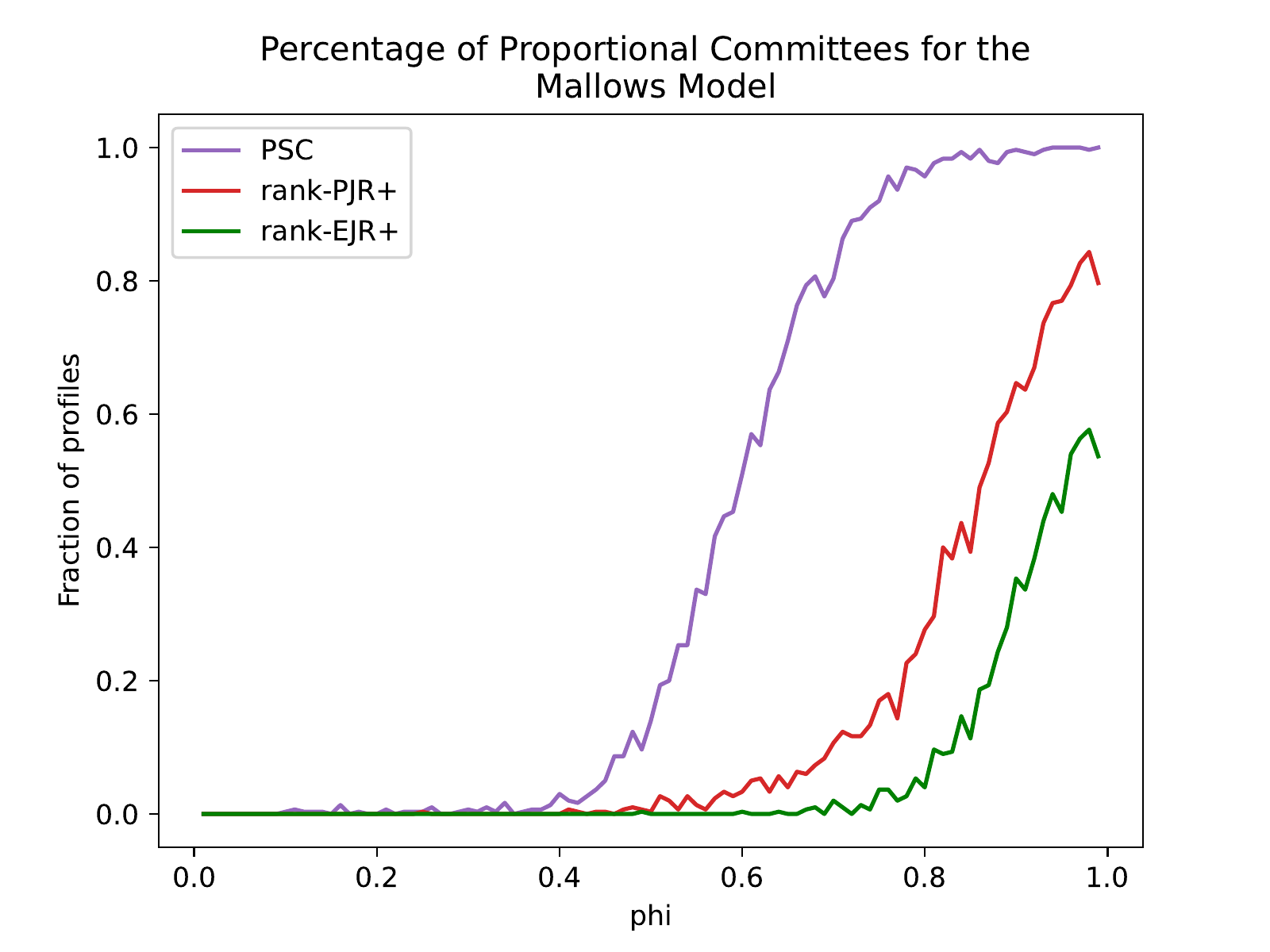}
\includegraphics[scale = \figsize]{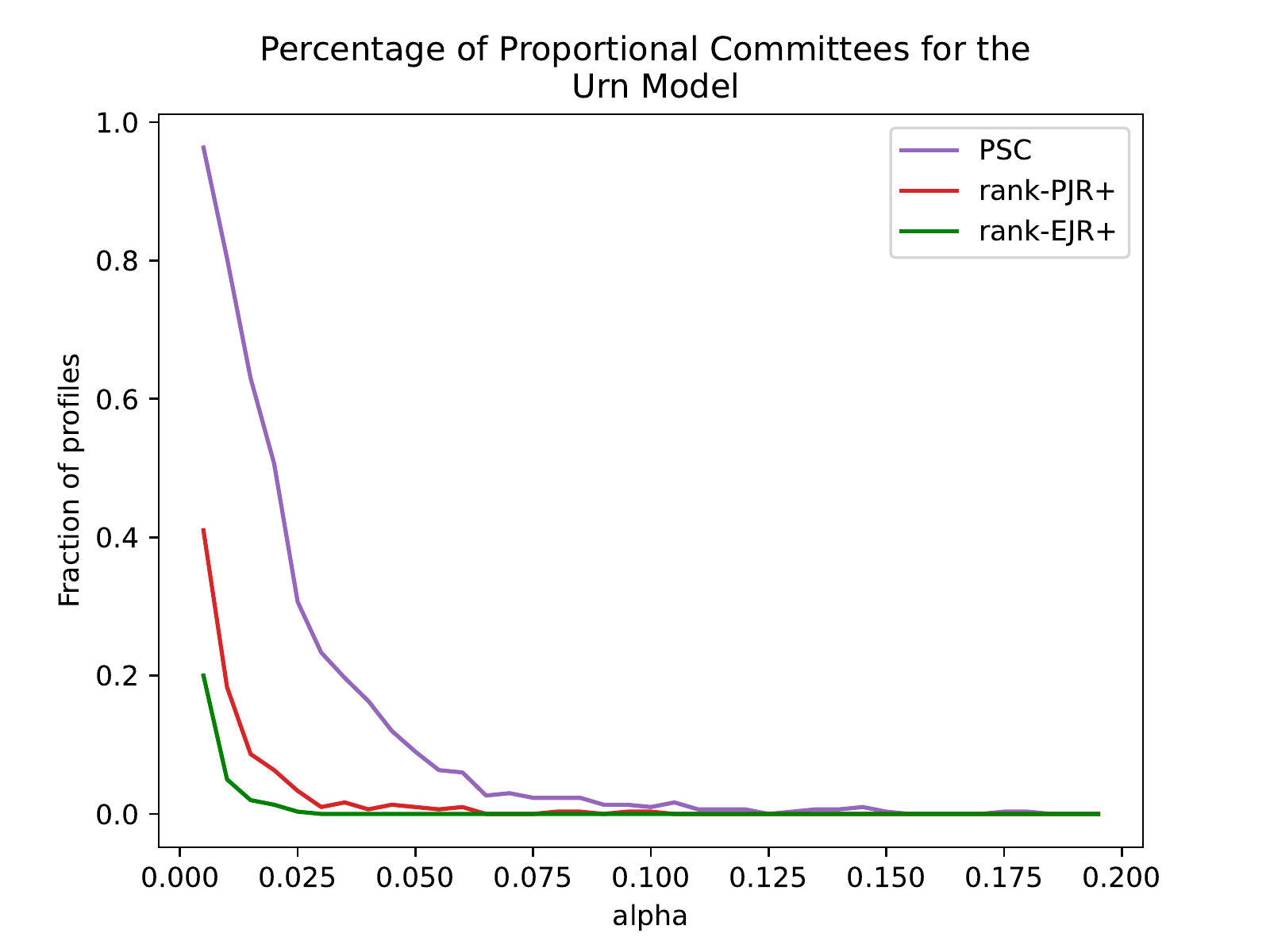}

\vspace{0.5cm}

\includegraphics[scale = \figsize]{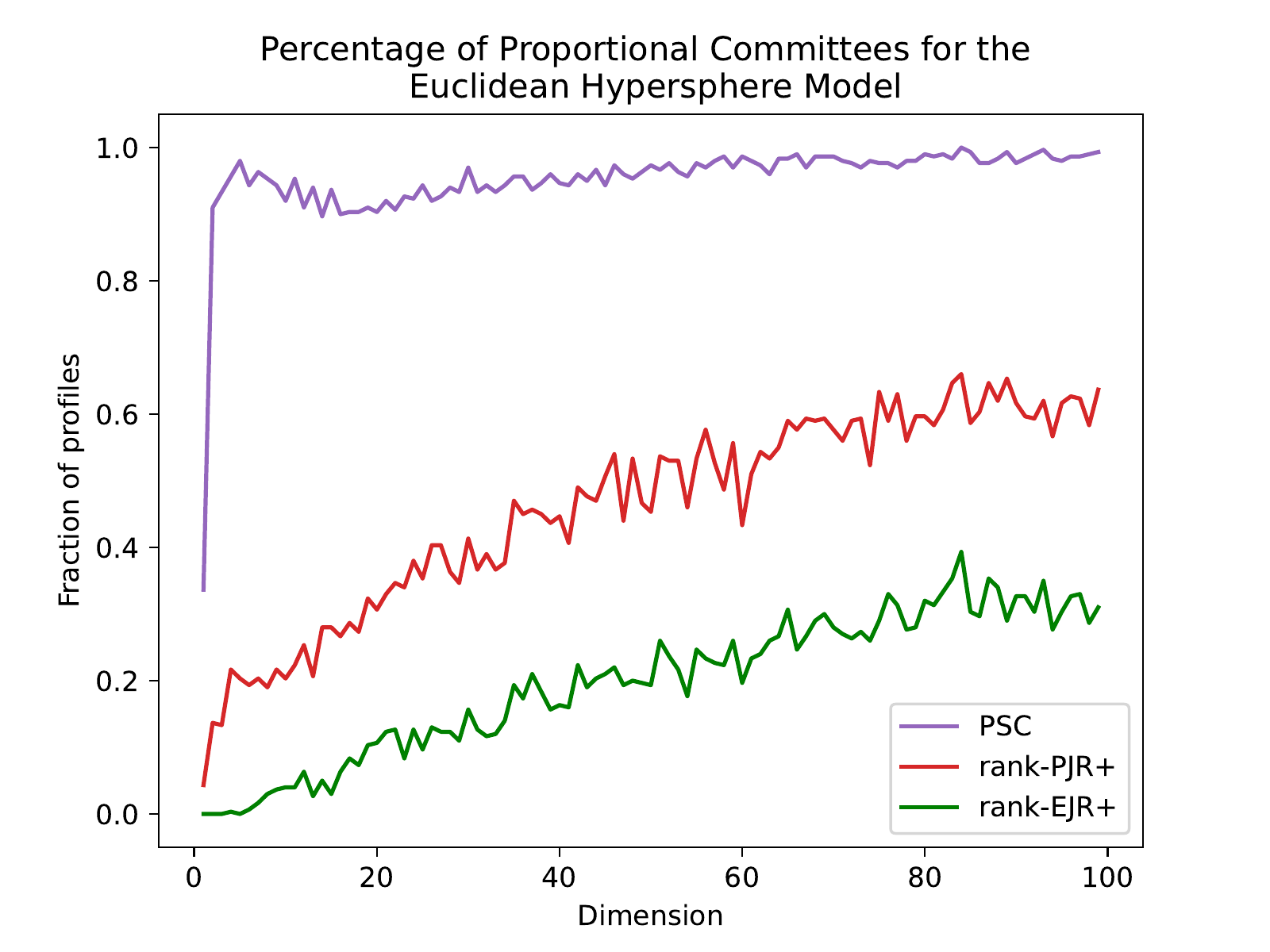}
\includegraphics[scale = \figsize]{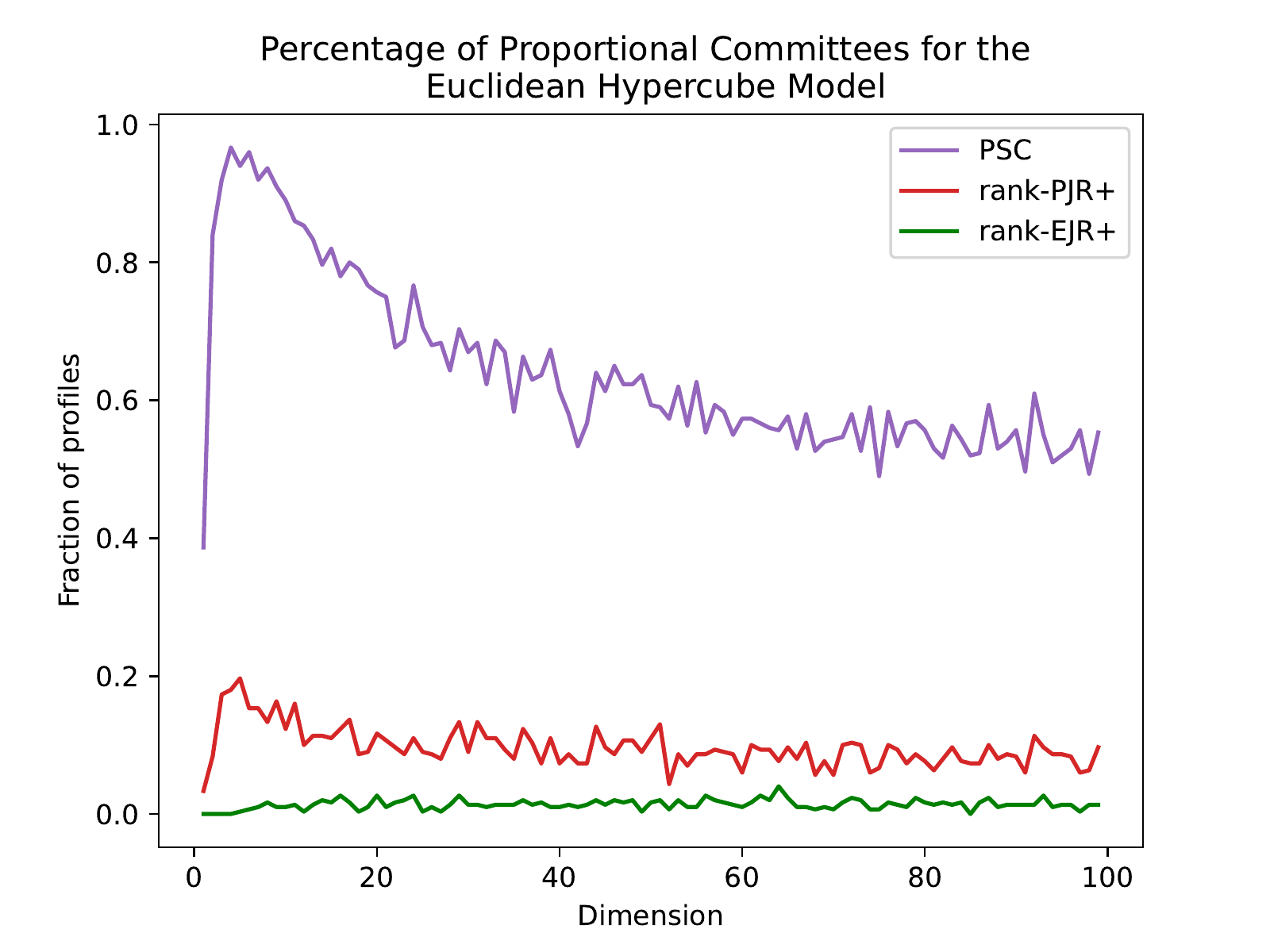}    
\caption{Experimental results for ranked ballots}
    \label{fig:ranked-experiments}
\end{figure}

\subsection{Experiments With Ranked Ballots}
\label{sec:exp-ranked}

As in the setting with approval preferences (\Cref{sec:exp-app}), we again ran experiments with randomly generated instances in order to compare the discriminative power of the proportionality axioms considered in this section. 

\paragraph{Setup}
We use a setting with $100$ voters and $50$ candidates and we generate strict rankings over all candidates using four different models: the classic Mallows model and urn model, as well as the Euclidean hypersphere and hypercube models, in which voters and candidates are uniformly distributed in a $d$-dimensional hypersphere or hypercube, with voters ranking candidates by distance. Detailed descriptions of these models can be found, e.g., in \citep{BBF+21a}.

For the Mallows model, we vary the $\phi$ parameter from $0$ to $1$; in the urn model, we vary the $\alpha$ parameter from $0$ to $0.2$; and in the Euclidean models, we vary the dimension from $1$ to $100$. For each parameter combination, we sample $300$ instances and one committee per instance, and we check whether this committee satisfies PSC, rank-\pjrp, and rank-\ejrp. 

\paragraph{Results}
The results are presented in \Cref{fig:ranked-experiments}. We find that both rank-\pjrp and rank-\ejrp are significantly more discriminating than PSC. Moreover, rank-\ejrp is harder to satisfy than rank-\pjrp (but committees satisfying the former axiom are not guaranteed to exist).

\clearpage

\section{Extensions}
\label{sec:extensions}

Finally, we demonstrate the value of our robustness approach by applying it to three related settings.

\subsection{Proportionality Degree}
\label{sec:prop-degree}

In its original definition \citep{Skow21a}, 
the proportionality degree measures how well cohesive groups are represented. Namely, for some function $f\colon \mathbb{N} \to \mathbb{R}$, a voting rule has a proportionality degree of $f$ if for every $\ell$-cohesive group $N'$, the voters in $N'$ approve, on average, at least $f(\ell)$ candidates in the committee. In order to make this definition more robust, we move away from arguments about cohesive groups and argue about unselected candidates instead.  
\begin{definition} \label{def:represented}
Given an instance with approval preferences and a committee $W$, a candidate $c \in C \setminus W$ is \emph{($f(\ell), \ell)$-represented} if there is no group of voters $N' \subseteq N_c$ with $\lvert N' \rvert \ge \frac{\ell n}{k}$ and 
\[
\frac{1}{\lvert N'\rvert }\sum_{i \in N'} \lvert A_i \cap W \rvert  < f(\ell). 
\]
A voting rule is $f$-representative for a function $f\colon \mathbb{N} \to \mathbb{R}$ if it only produces committees for which all unchosen candidates are $(f(\ell), \ell)$-represented.
\end{definition}

This is indeed a stronger notion than the proportionality degree.
\begin{restatable}{proposition}{propdegreerel}
 Any $f$-representative voting rule also has a proportionality degree of $f$.
\end{restatable}
\begin{proof}
Assume that we have a rule which is $f$-representative, but has a proportionality degree of less than $f$. Then there is a committee $W$ in the output of the rule, such that an $\ell$-cohesive group $N'$ exists with an average degree of less than $f(\ell)$. Since $f(\ell)$ can be at most $\ell$ 
one candidate $c$ from the group of candidates $N'$ is cohesive over, was not selected. However, then this candidate also witnesses that the committee $W$ is not $(f(\ell), \ell)$-representative, a contradiction.
\end{proof}

\citet{SFF+17a} have shown that any voting rule satisfying EJR has a proportionality degree of $\frac{\ell - 1}{2}$ and \citet{AEH+18a} have shown that PAV has a proportionality degree of $\ell-1$. 
We can show analogous results for \ejrp and representativeness.

\begin{restatable}{proposition}{propdegreel}
 Any voting rule satisfying \ejrp is $\frac{\ell - 1}{2}$-representative.
 \label{prop:ejrp_propdeg}
\end{restatable}
\begin{proof}
   The proof that EJR implies a proportionality degree of $\frac{\ell-1}2$ \citep{SFF+17a} does not use the cohesiveness of the voters approving the candidate outside the committee. Hence, the same proof also works to show the proposition. 
\end{proof}

\begin{restatable}{proposition}{propdegreemes}
    MES is $\frac{\ell-1}{2}$-representative and PAV is $\ell-1$-representative.
\end{restatable}
\begin{proof}
    The first statement follows from \Cref{prop:ejrp_propdeg}, the second again follows from the proof that PAV has a proportionality degree of $\ell-1$ \citep{LaSk22a}, since this proof does not use the cohesiveness of voter groups.
\end{proof}

An important advantage of the new notion is that we can efficiently check whether a given committee is $f$-representative.
\begin{restatable}{proposition}{propdegreever}
 Given an instance with approval preferences and a committee $W$, we can check in polynomial time whether $W$ is $f$-representative.
\end{restatable}
\begin{proof}
    For a group of voters $N'$ of size $\frac{\ell n}{k}$ witnessing a proportionality degree of at most $f(\ell)$ it is always sufficient to take the $\ell \frac{n}{k}$ voters approving $c$ with the least approved candidates in the committee, since the chosen candidates in $N'$ do not depend on each other.
\end{proof}

\noindent By contrast, the proportionality degree is coNP-complete to verify \citep{JaFa22a}.

\subsection{Participatory Budgeting}
\label{sec:pb}

Participatory Budgeting (PB) is a democratic innovation that enables citizens to decide on budget allocations of cities or districts \citep{Caba04a}. PB elections can be treated as a generalization of multiwinner voting in which the candidates (which are now referred to as \textit{projects)} have a \textit{cost} and the total cost of  selected projects cannot exceed a given budget limit. 

Several papers have suggested generalizations of proportionality axioms from multiwinner voting to the PB setting. Here, we focus on PB instances with approval preferences and we assume 
that the utility of a voter is given by the cost of the approved projects in the final allocation \citep{ALT18a, AzLe21a}.% 
\footnote{
Another common assumption is to measure the utility of a voter via the number of approved projects in the final allocation \citep{PPS21a, LCG22a}. Other utility functions have been considered by \citet{BFL+23a} and \citet{MREL23a}.}
Our robustness approach can be applied to notions like \textit{EJR/PJR up to one project} \citep{PPS21a} and \textit{EJR/PJR up to any project} \citep{BFL+23a}.

To facilitate the comparison of axioms, we state the definition of \textit{EJR up to any project} from the paper by \citet{BFL+23a}. 
\begin{definition}\label{def:pb-ejrx} 
    Given an approval-based PB instance, a committee $W$ satisfies \emph{EJR up to any project} (for cost utilities) if for every group $N' \subseteq N$ of voters  and group $T \subseteq \bigcap_{i \in N'} A_i$ of projects  such that $\sum_{p \in T} c(p) \le \frac{\lvert N'\rvert b}{n}$, there is a voter $i \in N'$ with 
    \[
    \sum_{p' \in A_i \cap W} c(p') + c(p) >  \sum_{p' \in T} c(p') \quad \text{ for all } p \in T.
    \]
\end{definition}
A robust strengthening of this axiom can be defined as follows. (Here, $c(p)$ denotes the cost of project $p$ and $B$ denotes the budget limit.)
\begin{definition}\label{def:pb-ejr+x}
    Given an approval-based PB instance, a feasible committee 
    $W$ satisfies \emph{\ejrp up to any project} (for cost utilities) if for every group $N' \subseteq N$ of voters and $p \in \bigcap_{i \in N'} A_i\setminus W$,
    there is an $i \in N'$ with $$c(A_i \cap W) + c(p) > \frac{\lvert N'\rvert b}{n}.$$
\end{definition}

The axiom in \Cref{def:pb-ejr+x} is stronger than EJR up to any project, but can be satisfied by an appropriate generalization of MES to the PB setting, using the approach of \citep{PPS21a}

\begin{definition}[MES for PB]
	Given a PB instance, MES starts with an empty committee $W = \emptyset$ and assigns a budget $b_i = \frac{b}{n}$ to each voter $i \in N$. A project $p \notin W$ is $\rho$-affordable if 
\[
\sum_{i \in N_p} \min(b_i, \rho c(p)) = c(p).
\]
It then selects the project $p_j$ which is $\rho$-affordable for the minimum $\rho$ and sets the budget $b_i$ is updated to 	$b_i - \min(b_i, \rho c(p))$ for every approver  $i$ of $p$.
MES continues until no $\rho$-affordable projects are left.
\end{definition}
This generalization now satisfies \ejrp up to any project.
\begin{restatable}{proposition}{propmes}
MES generalized to PB with cost utilities satisfies \ejrp up to any project.
\end{restatable}
\begin{proof} 
Assume that the output $W$ of MES does not satisfy \ejrp up to any project. Thus, there is a project $p$ and voters $N' \subseteq N_p$ with $c(A_i \cap W) + c(p) \le \frac{\lvert N' \rvert b}{n}$. Since $p$ is no longer affordable, we know that $\sum_{i \in N'} b_i < c(p)$.
Thus, we get that
\begin{align*}
    \frac{\sum_{i \in N'} \frac{b}{n} - b_i}{\sum_{i \in N'} c(A_i \cap W)} > \frac{\lvert N'\rvert \frac{b}{n} -c(p)}{\sum_{i \in N'} c(A_i \cap W)} \ge \frac{\lvert N'\rvert \frac{b}{n} -c(p)}{\sum_{i \in N'} \lvert N'\rvert \frac{b}{n} -c(p)} = \frac{1}{\lvert N'\rvert}.
\end{align*}
Hence, one voter must pay more than $\frac{1}{\lvert N'\rvert}$ per utility they receive and there must have been a candidate selected with a $\rho$ larger than $\frac{1}{\lvert N'\rvert}$. However, in the moment this candidate got selected, each voter in $N'$ still had less than $\frac{c(A_i \cap W)}{\lvert N'\rvert} \le \frac{b}{n} - \frac{c(p)}{\lvert N'\rvert}$ budget left. Thus, $p$ was still $\frac{1}{\lvert N'\rvert}$-affordable, a contradiction.
\end{proof}

We can also define a robust version of PJR up to any project.

\begin{definition}
Given an approval-based PB instance, a feasible committee $W$ satisfies \emph{\pjrp up to any project} (for cost utilities) if for any $N' \subseteq N$ and $p \in \bigcap_{i \in N'} A_i\setminus W$, it holds that $$c\left(\bigcup_{i \in N'} A_i \cap W\right) + c(p) > \frac{\lvert N'\rvert b}{n}.$$
\end{definition}
\pjrp up to any project would then in turn be implied by priceability (this statement can be proven analogously to \Cref{prop:rank-price}). This generalizes Theorem~4.4 by \citet{BFL+23a} and implies that the PB generalizations of MES and other rules such as Phragmén's sequential rule satisfy \pjrp up to any project.

\subsection{Querying Procedures} 
\label{sec:query}

Motivated by the application of civic participation platforms, \citet{HKP+23a} have recently introduced the problem of querying procedures for proportional committees. They study two models: (i) An \textit{exact query model}, in which one can query a subset of candidates $C' \subseteq C$ and, for each subset of $C'$, get the number of voters approving {exactly} this subset; and 
(ii) a \textit{noisy model}, 
where for a given queried subset $C' \subseteq C$, a single voter $i$ is selected uniformly at random and the set $A_i \cap C'$ is returned. 
Using these models, they showed how to simulate the \textit{Local Search-PAV} procedure by \citet{AEH+18a} to find a committee satisfying EJR 
using $\mathcal{O}(m k^2 \log(k))$ queries in the exact case and $\mathcal{O}(m k^6 \log(k)\log(m))$ queries with high probability in the randomized case. 

We show how the Greedy Justified Candidate Rule (\Cref{alg:gjcr}) can be employed to improve these bounds. 
The simplicity of the rule allows us to show that a committee satisfying \ejrp can be found using $\mathcal{O}(m \log(k))$ queries (of size at most $k$) in the exact model and with high probability using $\mathcal{O}(m k^4 \log(k)\log(m))$ in the noisy model.

We start with the exact setting and show that the Greedy Justified Candidate Rule (\Cref{alg:gjcr}) can be simulated using $\mathcal{O}(m \log(k))$ queries.
\begin{algorithm}[tb]
\caption{Noisy Greedy Justified Candidate Rule}
\label{alg:noisygjcr}

 $W \gets \emptyset$\;
 $h \gets 2\log(\frac{mk}{\delta}) (2k(k+1))^2$\;
 $\ell \gets k$\;
\While{$\ell \ge 1$}{
 $i \gets \lvert W \rvert$\;
 Partition $C \setminus W$ into $S_1, \dots, S_{\lceil \frac{m-i}{k-i}\rceil}$\;
 Query $W \cup S_t$ $h$ times for each $t \in [\lceil \frac{m-i}{k-i}\rceil]$\;
 Assign $a_c$ the number of queried voters $j \in N_c$ with $\lvert A_j \cap W \rvert < \ell$ for each $c \notin W$\;
\eIf{there is $c\notin W$: $\frac{a_c}{h} \ge \frac{\ell (2k + 1)}{2k(k+1)}$ }{
 $W \gets W \cup \{c\}$\;
}{
\While{there is no $c\notin W$: $\frac{a_c}{h} \ge \frac{\ell (2k + 1)}{2k(k+1)}$}{$\ell \gets  \ell - 1$\;}
}
}
 return $W$\;
\end{algorithm}
\begin{proposition}
 The Greedy Justified Candidate Rule can be implemented using $\mathcal{O}(m \log(k))$ queries of size at most $k$.
\end{proposition}
\begin{proof}
 We show that it can be implemented in $\mathcal{O}(m \log(k))$ queries. Let $\{c_1, \dots, c_i\}$ be the already selected candidates. Then we can partition the unselected $m-i$ candidates into $\lceil \frac{m-i}{k-i}\rceil$ sets $S_1, \dots, S_{\lceil \frac{m-i}{k-i}\rceil}$ of size at most $k-i$. By querying every $S_j \cup \{c_1, \dots, c_i\}$, we can determine the next candidate who would be chosen by \gjcr.
 Thus, we only need to query \begin{align*}
     \sum_{i = 0}^{k-1} \left\lceil \frac{m-i}{k-i} \right\rceil \le k + m  \sum_{i = 0}^{k-1} \frac{1}{k-i} \in \mathcal{O}(m \log(k))
 \end{align*} many sets.
\end{proof}

Next, we show how \gjcr can be adapted to the noisy model as \Cref{alg:noisygjcr}.

\begin{proposition}
    For any $\delta > 0$, \Cref{alg:noisygjcr} with probability $1-\delta$ and $\mathcal{O}(m k^4 \log(m)\log(k))$ queries returns a committee of size at most $k$ satisfying \ejrp.
\end{proposition}
 \begin{proof}
 For a given $\ell$, we call a candidate $c$ ``large'' if $\lvert \{i \in N_c \colon \lvert A_i \cap W \rvert < \ell\} \rvert \ge \frac{\ell n}{k}$ and ``tiny'' if $\lvert \{i \in N_c \colon \lvert A_i \cap W \rvert < \ell\} \rvert \le \frac{\ell n}{k+1}$. Our goal is now to show that during the whole process the probability that a tiny candidate is selected or that $\frac{a_c}{h} < \frac{\ell 2k + 1}{2k(k+1)}$ for a large candidate is at most $\delta$. This way, we can already ensure that the chosen committee satisfies \ejrp with probability $1-\delta$ and that the committee is of size at most $k$.
 
 We first bound the probability that a tiny candidate is chosen. Since there are $m$ candidates and at most $k$ candidates get chosen, it is sufficient to upper bound this by $\frac{\delta}{4mk}$ by the union bound to show a general bound of $\frac{\delta}{4}$. Let $c$ be any tiny candidate for a given $\ell$. Since in any round we have that $\mathbb{E}[a_c] \le \frac{\ell}{k+1}h$ we get that $\mathbb{P}\left[a_c \ge h\frac{\ell (2k + 1)}{2k(k+1)}\right] = \mathbb{P}\left[a_c -  \mathbb{E}[a_c]\ge h(\frac{\ell (2k + 1)}{2k(k+1)} - \frac{\ell}{(k+1)})\right] = \mathbb{P}\left[a_c -  \mathbb{E}[a_c]\ge h(\frac{\ell}{2k(k+1)})\right]$. Thus, using the Chernoff-Hoeffding inequality we get that 
 \[\mathbb{P}\left[a_c \ge h\frac{\ell (2k + 1)}{k(k+1)}\right] \le \exp\left(-2 (\frac{\ell}{2k(k+1)})^2 h\right) = \exp\left(-2 \frac{2}{\ell^2} \log(\frac{mk}{\delta})\right) \le \frac{\delta}{4mk}.\] 
 Hence, the probability that a small candidate is selected is at most $\frac{\delta}{4}$. 
 
 Similarly, we can bound the probability of a large candidate not being selectable by 
 \begin{align*}
 \mathbb{P}&\left[a_c \le h\frac{\ell (2k + 1)}{k(k+1)}\right] =  \mathbb{P}\left[\mathbb{E}[a_c] - a_c \ge h(\frac{\ell}{k} - \frac{\ell (2k + 1)}{k(k+1)})\right] \\
 &= \mathbb{P}\left[\mathbb{E}[a_c] - a_c \ge h(\frac{\ell}{k} - \frac{\ell (2k + 1)}{k(k+1)})\right] = \mathbb{P}\left[\mathbb{E}[a_c] - a_c \ge h(\frac{\ell}{2k(k+1)})\right].
  \end{align*}
 Similarly to earlier, we can also upper bound this probability by $\frac{\delta}{4mk}$. Thus, the probability that a large candidate is unselected and a small candidate is selected is at most $\delta$. Hence, since only candidates of with strictly more than $\frac{\ell n}{k+1}$ voters get selected, this means that the committee size is at most $k$ and satisfies \ejrp, since otherwise a large candidate would have been unselected.

 Thus, since at most $k$ candidates get selected, we use $h \mathcal{O}(m 
 \log(k)) = \mathcal{O}(m k^4 \log(m)\log(k))$ queries in total.
 \end{proof}

\newcommand{\yes}{\textcolor{green!50!black}{\ding{52}}}
\newcommand{\no}{\textcolor{red}{\ding{55}}}

\tikzset{newaxiom/.style={rectangle, text width=2.5cm,minimum height=0.3cm, draw, align=center, double, double distance=2pt, outer sep=2pt}}
\tikzset{newaxiom2/.style={rectangle, text width=3cm,minimum height=0.3cm, draw, align=center, double, double distance=2pt, outer sep=2pt}}
\tikzset{diagnode/.style={rectangle, text width=2.5cm,minimum height=0.3cm, draw, align=center, outer sep=1.2pt}
}
\tikzset{diagnode2/.style={rectangle, text width=3cm,minimum height=0.3cm, draw, align=center, outer sep=1.2pt}
}
\tikzset{easynode/.style={rectangle, text width=2.5cm,minimum height=0.3cm, draw, align=center, color = green!50!black}
}
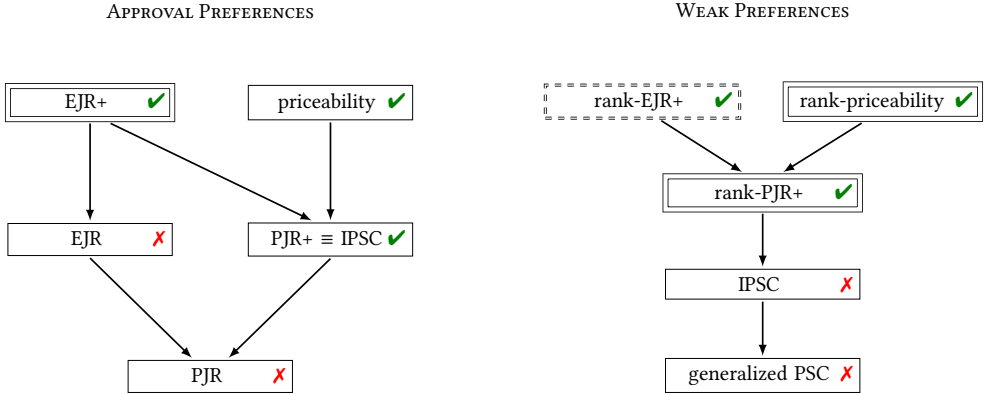
\begin{figure}
\scalebox{0.8}{
\begin{tikzpicture}[yscale=0.75]
\node at (2,8) {\textsc{Approval Preferences}};
\node at (0,6) [newaxiom](ejrp) {\mbox{}\hfill \ejrp\hfill\mbox{}\makebox[0pt][r]{\yes}};

\node at (4,6) [diagnode](price) {\mbox{}\hfill priceability \hfill\mbox{}\makebox[0pt][r]{\yes}};

\node at (0,3) [diagnode](ejr) {\mbox{}\hfill EJR \hfill\mbox{}\makebox[0pt][r]{\no}};

\node at (4,3) [diagnode](pjrp) {\mbox{}\hfill \pjrp $\equiv$ IPSC \hfill\mbox{}\makebox[0pt][r]{\yes}};

\node at (2,0) [diagnode](pjr) {\mbox{}\hfill PJR \hfill\mbox{}\makebox[0pt][r]{\no}};

\draw[-latex, thick] (ejrp.315) -- (pjrp.135);
\draw[-latex, thick] (ejr.south) -- (pjr.135);
\draw[-latex, thick] (price.south) -- (pjrp.north);
\draw[-latex, thick] (ejrp.south) -- (ejr.north);
\draw[-latex, thick] (pjrp.south) -- (pjr.45);
\end{tikzpicture}

\hspace{2cm}

\begin{tikzpicture}[yscale=0.75]
\node at (2,8) {\textsc{Weak Preferences}};
\node at (0,6) [diagnode2, dashed, double, outer sep=1.2pt](ejrp) {\mbox{}\hfill rank-\ejrp\hfill\mbox{}\makebox[0pt][r]{\yes}};

\node at (4,6) [newaxiom2](price) {\mbox{} rank-priceability  \hfill\mbox{}\makebox[0pt][r]{\yes}};

\node at (2,4) [newaxiom2](pjrp) {\mbox{}\hfill {rank-\pjrp } \hfill\mbox{}\makebox[0pt][r]{\yes}};

\node at (2,2) [diagnode2](pjr) {\mbox{}\hfill IPSC \hfill\mbox{}\makebox[0pt][r]{\no}};

\node at (2,0) [diagnode2](psc) {\mbox{}\hfill {generalized PSC} \hfill\mbox{}\makebox[0pt][r]{\no}};

\draw[-latex, thick] (ejrp.315) -- (pjrp.135);
\draw[-latex, thick] (price.225) -- (pjrp.45);
\draw[-latex, thick] (pjr.south) -- (psc.north);
\draw[-latex, thick] (pjrp.south) -- (pjr.north);
\end{tikzpicture}
}
    \caption{Relations between the proportionality axioms considered in this paper. Nodes with a double border correspond to axioms that have been proposed in this paper. For all axioms other than rank-\ejrp, committees satisfying the axiom are guaranteed to exist and can be found efficiently.
    Axioms marked with ``\yes\xspace'' can also be verified efficiently, whereas a ``\no\xspace'' indicates that the verification problem is computationally intractable.}
    \label{fig:relations}
\end{figure}

\section{Conclusion}

We have proposed novel proportionality axioms for multiwinner voting, both in the approval-based and in the ranking-based setting. Figure \ref{fig:relations} summarizes the relations between the proportionality axioms considered in this paper.
Our axioms are more robust and easier to verify than existing ones. Moreover, committees satisfying these axioms are guaranteed to exist and can be found by applying polynomial-time computable voting rules that have other attractive properties \citep{AzLe20a,PeSk20a,PPS21a}.  

We have demonstrated that our approach can also be used to ``robustify'' the proportionality degree and proportionality axioms for participatory budgeting. 
The computational benefits of our approach go beyond verifiability: The simple structure of \ejrp enables the optimization over all \ejrp committees and gives rise to a simple and very useful greedy procedure. 

Our work gives rise to multiple follow-up questions. First, it would be interesting to see whether a robust version of \textit{fully proportional representation (FJR)} \citep{PPS21a} can be defined. Finding such a generalization might prove useful in answering the open question whether committees satisfying FJR can be found in polynomial time. 
Second, our extensions to participatory budgeting (\Cref{sec:pb}) only apply to cost utilities. Are there similar notions which also apply to other utility notions \citep{BFL+23a}, or maybe even to general additive utilities? Third, it would be interesting to see if there is a generalization of \ejrp to the ranked setting that is always satisfiable.

\begin{acks}
We thank Haris Aziz, Piotr Faliszewski, Martin Lackner, Jan Maly, Dominik Peters, Grzegorz Pierczyński, Piotr Skowron, Warut Suksompong, and Bill Zwicker for helpful discussions.  

This research is supported by the Deutsche Forschungsgemeinschaft (DFG) under the grant BR~4744/2\nobreakdash-1 and the Graduiertenkolleg ``Facets of Complexity'' (GRK~2434). 
\end{acks}

    \bibliographystyle{abbrvnat}
\bibliography{abb,algo,bibliography}

\newpage

\appendix

\section{Example of Aziz and Lee}
\label{app:AzLe}

\citet{AzLe20a} considered the following example to demonstrate a weakness of PSC.  
\begin{align*}
    1 \colon c_1 \succ c_2 \succ c_3 \succ e_1 \succ e_2 \succ e_3 \succ e_4 \succ d_1 \\
    2 \colon c_2 \succ c_3 \succ c_1 \succ e_1 \succ e_2 \succ e_3 \succ e_4 \succ d_1 \\
    3 \colon c_3 \succ c_1 \succ d_1 \succ c_2 \succ e_1 \succ e_2 \succ e_3 \succ e_4 \\
    4 \colon e_1 \succ e_2 \succ e_3 \succ e_4 \succ c_1 \succ c_2 \succ c_3 \succ d_1 \\
    5 \colon e_1 \succ e_2 \succ e_3 \succ e_4 \succ c_1 \succ c_2 \succ c_3 \succ d_1 \\
    6 \colon e_1 \succ e_2 \succ e_3 \succ e_4 \succ c_1 \succ c_2 \succ c_3 \succ d_1 \\
    7 \colon e_1 \succ e_2 \succ e_3 \succ e_4 \succ c_1 \succ c_2 \succ c_3 \succ d_1 \\
    8 \colon e_1 \succ e_2 \succ e_3 \succ e_4 \succ c_1 \succ c_2 \succ c_3 \succ d_1 \\
    9 \colon e_1 \succ e_2 \succ e_3 \succ e_4 \succ c_1 \succ c_2 \succ c_3 \succ d_1
\end{align*}
There are  $n = 9$ voters and $k=3$, so that $\frac{n}{k} = 3$.
In this instance, STV could select the committee $\{e_1, e_2, e_3\}$ by first selecting $e_1, e_2$, then eliminating $c_1, c_2, c_3, d_1$ and finally selecting $e_3$. This committee does not satisfy \rpjrp, since either $c_1, c_2$ or $c_3$ witness a violation with the first three voters. Indeed, the only committees satisfying \rpjrp in this instance are $\{e_1, e_2\}$ together with one of $\{c_1, c_2, c_3, d_1\}$. EAR here selects either $\{e_1, e_2, c_1\}$ or $\{e_1, e_2, c_3\}$.

\end{document}